\documentclass[5p,times,number]{elsarticle}

\usepackage[T1]{fontenc}
\usepackage[utf8]{inputenc}
\usepackage{microtype}
\usepackage{amsmath,amssymb,amsthm,mathtools}
\usepackage{booktabs,tabularx,array}
\usepackage{float}
\usepackage{placeins}
\usepackage{flushend}
\usepackage{algorithm}
\usepackage[noend]{algpseudocode}
\usepackage{enumitem}
\usepackage{xcolor}
\usepackage{graphicx}
\usepackage{multirow}
\usepackage{balance}

\usepackage{tikz}
\usetikzlibrary{arrows.meta,backgrounds,calc,fit,positioning,shapes.geometric}
\usepackage{hyperref}
\usepackage[nameinlink,noabbrev]{cleveref}
\biboptions{sort&compress}
\makeatletter
\providecommand{\theHALG@line}{\thealgorithm.\arabic{ALG@line}}
\makeatother

\definecolor{qbpnavy}{HTML}{17324D}
\definecolor{qbpblue}{HTML}{2A6FBB}
\definecolor{qbpteal}{HTML}{168C88}
\definecolor{qbpgreen}{HTML}{5A9E45}
\definecolor{qbporange}{HTML}{E8892D}
\definecolor{qbpred}{HTML}{C94A4A}
\definecolor{qbpgray}{HTML}{65727E}
\definecolor{qbplight}{HTML}{EEF4F7}

\makeatletter
\newcommand{\qbpMarkRedBibEntry}[1]{%
  \expandafter\gdef\csname qbp@redbib@#1\endcsname{}}

\let\qbp@original@bibitem\bibitem
\newcommand{\qbp@setbibcolour}[1]{%
  \ifcsname qbp@redbib@#1\endcsname
    \color{red}\hypersetup{urlcolor=red}%
  \else
    \color{black}\hypersetup{urlcolor=qbpblue}%
  \fi}
\def\qbp@bibitem@plain#1{%
  \qbp@setbibcolour{#1}\qbp@original@bibitem{#1}}
\def\qbp@bibitem@optional[#1]#2{%
  \qbp@setbibcolour{#2}\qbp@original@bibitem[#1]{#2}}
\def\bibitem{\@ifnextchar[\qbp@bibitem@optional\qbp@bibitem@plain}
\makeatother

\hypersetup{
  colorlinks=true,
  linkcolor=qbpnavy,
  citecolor=qbpteal,
  urlcolor=qbpblue,
  pdftitle={Graph-Aware Exact Branch-and-Bound with Device Profiles for Static Qubit Allocation}
}
\setlist[itemize]{leftmargin=*,topsep=3pt,itemsep=2pt}

\newtheorem{definition}{Definition}
\newtheorem{proposition}{Proposition}

\theoremstyle{remark}

\newcommand{\QAP}{\textnormal{\textsc{QAP}}}
\newcommand{\GLB}{\textnormal{\textsc{GLB}}}
\newcommand{\HHB}{\textnormal{\textsc{HHB}}}
\newcommand{\LAP}{\textnormal{\textsc{LAP}}}
\newcommand{\cfg}[1]{\texttt{#1}}
\newcommand{\Aut}{\operatorname{Aut}}
\newcommand{\lap}{\operatorname{LAP}}

\newcommand{\Fb}{\mathbf{F}}
\newcommand{\Db}{\mathbf{D}}
\newcommand{\Tb}{\mathbf{T}}
\newcommand{\Ab}{\mathbf{A}}
\newcommand{\Gb}{\mathbf{G}}

\journal{Future Generation Computer Systems}

\begin{document}

\begin{frontmatter}

\title{{Graph-Aware Exact Branch-and-Bound with Device Profiles for Static Qubit Allocation}}

% AUTHOR ACTION REQUIRED BEFORE SUBMISSION:
% Replace this draft author/affiliation block with entries of the form
 \author[aff1,aff2]{Kamer Kaya\corref{cor1}}
 \ead{kaya@sabanciuniv.edu}
 \cortext[cor1]{Corresponding author}
 \affiliation[aff1]{organization={Faculty of Engineering and Natural Sciences},
   addressline={Sabanc{\i} University}, city={İstanbul}, postcode={34956},
   state={Tuzla}, country={Türkiye}}
  \affiliation[aff2]{organization={Center of Excellence in Data Analytics (VERİM)},
   addressline={Sabanc{\i} University}, city={İstanbul}, postcode={34956},
   state={Tuzla}, country={Türkiye}
}

\begin{abstract}
Static qubit allocation maps a circuit's logical qubits to a sparse physical
device while minimising an interaction-weighted physical-distance cost function, yielding a
rectangular quadratic assignment problem.  Existing work combines strong
lower bounds with distributed branch-and-bound. We integrate graph-aware exact reductions with an
engineering bundle for a lightweight assignment-bound path:
unavoidable assigned-cost filtering, incrementally maintained root-orbit and prefix-stabilizer symmetry pruning, conditioned parent-\LAP{} screening, and circuit-independent physical device profiles. On 21 relatively easy {\tt Melbourne} instances and six {\tt Boeblingen}
instances completed by the \GLB{} baseline, the final single-thread configuration
provides geometric-mean speedups of \(2.98\times\) and \(13.27\times\),
respectively. With 60 threads on one shared-memory
server, all instances in the final {\tt Boeblingen--Cairo} experiment are certified optimal within half an hour, excluding one-time device-artifact
construction. These results show that graph-aware node processing and engineering the search process substantially reduce the resources required for exact allocation.
\end{abstract}

\begin{keyword}
qubit allocation \sep quadratic assignment \sep branch-and-bound \sep
graph automorphism \sep assignment lower bound 
\end{keyword}

\end{frontmatter}

\section{Introduction}
\label{sec:introduction}

Quantum device connectivity makes the location of logical qubits an important decision. 
Before routing a circuit, an allocator
selects a mapping from logical (soft) qubits into physical (hard) qubits.  
A poor choice
can force long interaction paths and inefficient use of resources; an exact choice is useful as a reference
for heuristic mappers and as a way to understand which device properties
make an instance difficult. The difficulty is combinatorial: assigning
\(n\) logical qubits to \(N\geq n\) physical qubits gives
\(N!/(N-n)!\) possible injections.

Valois \emph{et al.} formulate this problem as a quadratic assignment and solve it with an exact branch-and-bound algorithm using a four-index dual-concentration procedure~\citep{hahn1998,hahn1998bnb} in a distributed framework \citep{valois2026,helbecque2023}.  Their work provides the natural
baseline for this study.  We do not replace its 
problem formulation: instead, we ask an algorithms-and-systems
question: Can information already present in a fixed device graph and in the partial allocations encountered during branch-and-bound safely eliminate work before evaluating expensive bounds?

This study devises a two-part answer to this problem. First, the {\em combinatorial bundle} uses inexpensive graph-aware tests to reduce the work performed before more expensive bounds are evaluated. For instance, the cost already fixed by a partial 
allocation is a valid lower bound and is cheap to update. Furthermore, automorphisms of 
the physical device graph identify equivalent placements, and after each assignment, the subgroup that
still fixes the assigned physical qubits gives additional \emph{dynamic
symmetry}. In addition, one residual assignment certificate can evaluate the parent \LAP{} conditioned on each possible next assignment, allowing the to-be-pruned children detected before their \GLB{} matrices are constructed. Second, because the device graph is fixed across circuits,
an \emph{engineering bundle} reuses node state, assignment
certificates, and device-side data without changing the branch-and-bound tree.  The proposed approach enumerates the binary free-set masks once, stores the distance profile of every candidate vertex in that mask, and reuses the artifact across all circuits. We are not aware of prior exact static-allocation work that materialises every free-set/candidate distance histogram and uses them for constant-time assembly of the \GLB{} row-relaxation term. Our contributions are as follows.
\begin{itemize}
  \item We combine assigned-cost filtering and conditioned
parent-\LAP{} screening with an incremental group-based implementation
of the classical physical-location symmetry reduction, and separately
evaluate its root-only \(S_0\) and prefix-dependent \(S_\ast\) forms in our {\em combinatorial bundle}.
  
  \item As a part of the {\em engineering bundle}, we introduce a circuit-independent device preprocessor. It visits the \(2^N\) physical sets, compresses
  their candidate distance histograms, and gives an exact constant-time
  lookup for the row-relaxation term used in each \GLB{} matrix entry. The final code is publicly available\footnote{\url{https://github.com/kamerkaya/StaticQubitAllocation}}.
  
  \item We evaluate the proposed techniques through an oracle-cutoff ablation, practical heuristic-seeded single-thread runs, strong-scaling runs up to 32 threads, and a 60-thread final experiment.
  
  \item The final configuration on one 60-core shared-memory machine certifies optimality for all 22 instances in the final {\tt Boeblingen--Cairo} experiment within 29.6 minutes, whereas the previous work requires 6,795 seconds for the hardest instance on 64 compute nodes, each with 128 CPU cores, albeit on a different server architecture.
\end{itemize}

The remainder of the paper is organized as follows. \Cref{sec:background}
defines the allocation problem and the search state. \Cref{sec:bounds}
derives the graph-aware operators and the final configurations.
\Cref{sec:experiments} presents the sequential, heuristic, and shared-memory
measurements.  \Cref{sec:related} positions the work against exact
\Cref{sec:conclusion} concludes and outlines the next steps.

\section{Background and notation}
\label{sec:background}

\subsection{Static qubit allocation}

A gate-model \emph{quantum circuit} is read from left to right, with one
horizontal wire for each logical qubit.  A one-qubit gate acts on one wire, whereas a \emph{controlled-NOT (CNOT)} joins a filled control to a target \(\oplus\).  \Cref{fig:circuit-primer} collects this notation and shows how the temporal gate sequence is reduced to the weighted interaction graph used by the allocation model.  Single-qubit gates and final measurements do not enter the interaction matrix.

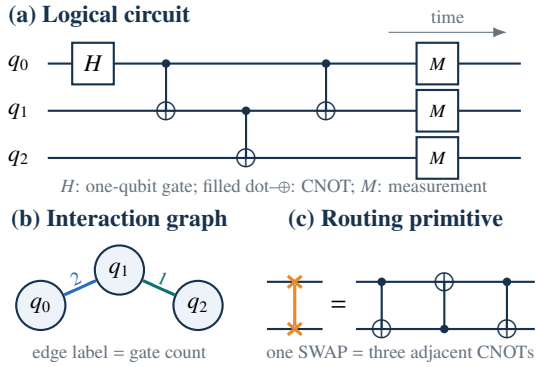
\begin{figure}[htbp]
\centering
\begin{tikzpicture}[
  x=0.92cm,y=0.78cm,>=Latex,
  wire/.style={draw=qbpnavy,thick},
  gate/.style={draw=qbpnavy,thick,fill=white,minimum size=5.5mm,
    inner sep=0pt,font=\small\bfseries},
  qvertex/.style={circle,draw=qbpnavy,thick,fill=qbpblue!8,
    minimum size=6.5mm,inner sep=0pt,font=\small},
  panel/.style={font=\small\bfseries,text=qbpnavy},
  note/.style={font=\scriptsize,text=qbpgray,align=center}
]
  % (a) A circuit read from left to right.
  \node[panel,anchor=west] at (0.05,6.55) {(a) Logical circuit};
  \foreach \y/\q in {5.75/0,4.95/1,4.15/2}{
    \draw[wire] (0.75,\y)--(7.55,\y);
    \node[font=\small,anchor=east] at (0.62,\y) {$q_\q$};
  }
  \node[gate] at (1.40,5.75) {$H$};
  \fill[qbpnavy] (2.45,5.75) circle (1.6pt);
  \node[font=\large,text=qbpnavy] at (2.45,4.95) {$\oplus$};
  \draw[wire] (2.45,5.75)--(2.45,4.95);
  \fill[qbpnavy] (3.60,4.95) circle (1.6pt);
  \node[font=\large,text=qbpnavy] at (3.60,4.15) {$\oplus$};
  \draw[wire] (3.60,4.95)--(3.60,4.15);
  \fill[qbpnavy] (4.75,5.75) circle (1.6pt);
  \node[font=\large,text=qbpnavy] at (4.75,4.95) {$\oplus$};
  \draw[wire] (4.75,5.75)--(4.75,4.95);
  \foreach \y in {5.75,4.95,4.15}{
    \node[gate,font=\scriptsize] at (6.35,\y) {$M$};
  }
  \draw[-{Latex[length=1.8mm]},qbpgray]
    (5.65,6.28)--(7.35,6.28) node[midway,above,note]{time};
  \node[note] at (4.00,3.64)
    {$H$: one-qubit gate; filled dot--$\oplus$: CNOT; $M$: measurement};

  % (b) The two-qubit gates aggregated into an interaction graph.
  \node[panel] at (1.78,3.10) {(b) Interaction graph};
  \node[qvertex] (g0) at (0.65,1.65) {$q_0$};
  \node[qvertex] (g1) at (1.78,2.25) {$q_1$};
  \node[qvertex] (g2) at (2.91,1.65) {$q_2$};
  \draw[very thick,qbpblue] (g0)--node[above,sloped,font=\scriptsize,
    fill=white,inner sep=1pt] {$2$} (g1);
  \draw[very thick,qbpteal!80!black] (g1)--node[above,sloped,
    font=\scriptsize,fill=white,inner sep=1pt] {$1$} (g2);
  \node[note] at (1.78,0.83) {edge label = gate count};

  % (c) The routing primitive that motivates short physical distances.
  \node[panel] at (5.75,3.10) {(c) Routing primitive};
  \draw[wire] (3.90,2.05)--(4.70,2.05);
  \draw[wire] (3.90,1.25)--(4.70,1.25);
  \draw[very thick,qbporange]
    (4.20,1.95)--(4.40,2.15)
    (4.20,2.15)--(4.40,1.95)
    (4.20,1.15)--(4.40,1.35)
    (4.20,1.35)--(4.40,1.15)
    (4.30,2.00)--(4.30,1.30);
  \node[font=\normalsize] at (4.93,1.65) {$=$};
  \draw[wire] (5.18,2.05)--(7.75,2.05);
  \draw[wire] (5.18,1.25)--(7.75,1.25);
  \foreach \x/\yc/\yt in {
    5.55/2.05/1.25,6.45/1.25/2.05,7.35/2.05/1.25}{
    \fill[qbpnavy] (\x,\yc) circle (1.6pt);
    \node[font=\large,text=qbpnavy] at (\x,\yt) {$\oplus$};
    \draw[wire] (\x,\yc)--(\x,\yt);
  }
  \node[note] at (5.82,0.83) {one SWAP $=$ three adjacent CNOTs};
\end{tikzpicture}
\caption{Circuit information used by static allocation.  (a) Gates are
applied from left to right.  (b) Repeated two-qubit gates are aggregated into
weighted logical edges; single-qubit gates and measurements do not contribute
to \(\Fb\).  (c) A SWAP decomposes into three CNOTs, motivating allocations
that place frequent interaction partners close together.}
\label{fig:circuit-primer}
\end{figure}

Let \({\cal{L}}=\{0,\ldots,n-1\}\) be the set of logical qubits and let
\({\cal{P}}=\{0,\ldots,N-1\}\), with \(N\geq n\), be the physical qubits.  The
device is an undirected graph \(G_\Db=({\cal P},{\cal E})\).  The hop distance between two vertices is denoted
by \(h(p,p')\).  Following the convention used by
Valois~\emph{et al.}, the routing-distance matrix is
\begin{equation}
 {\Db}_{p,p'}=\begin{cases}
 0, & p=p',\\
 h(p,p')-1, & p\ne p'.
 \end{cases}
 \label{eq:device-distance}
\end{equation}
Thus, directly connected physical qubits have distance zero.

Let \(\Fb_{i,j}\geq0\) be the matrix recording how often logical qubits \(i\) and \(j\) in a circuit interact. We use the unordered interaction weight
\(w_{i,j}=\Fb_{i,j}+\Fb_{j,i}\) for \(i \neq j\), and take \(w_{i,i}=0\). Thus a gate is counted once in the $(q_i,q_j)$ label of the circuit, although its count appears in two matrix positions $\Fb_{i,j}$ and $\Fb_{j,i}$. An \emph{allocation} is an injection
\(\pi:{\cal L}\rightarrow {\cal P}\).  Its static distance-surrogate cost is
\begin{equation}
 Q(\pi)=\sum_{0\leq i<j<n}w_{i,j}\Db_{\pi(i),\pi(j)}.
 \label{eq:objective}
\end{equation}
Note that~\eqref{eq:objective} defines the static distance-surrogate problem studied in the literature, e.g.,~\cite{valois2026}. It does not model placement changes or route selection by a dynamic router; heuristic mappers address that richer setting~\citep{li2019sabre}.

\begin{definition}[Static qubit allocation]
\label{def:allocation}
Given a logical interaction matrix \(\Fb\) and a physical distance matrix \(\Db\),
find an allocation \(\pi:{\cal L}\rightarrow {\cal P}\) that minimises
\(Q(\pi)\) in \eqref{eq:objective}.
\end{definition}

The toy logical circuit graph in \Cref{fig:toy} has three
qubits, and its edge labels use raw pair-count convention as
\Cref{fig:circuit-primer}.  The physical device is a four-cycle.  Although deliberately
small, it shows the three ideas that matter later: a partial mapping has
unavoidable cost, the device has symmetries, and a residual assignment can
evaluate a forced next choice.  For example, if a branch assigns
\(q_0\mapsto p_0\) and \(q_1\mapsto p_2\), then since
\(\Db_{p_0,p_2}=1\), the pair contributes
\(2 \times 3 \times 1=6\) to the objective.

\begin{figure}[htbp]
\centering
\scalebox{0.9}{
\begin{tikzpicture}[
  >=Latex,
  logical/.style={circle,draw=qbpnavy,very thick,fill=qbpblue!8,
    minimum size=8.5mm,inner sep=0pt,font=\normalsize},
  physical/.style={circle,draw=qbpgray,thick,fill=white,
    minimum size=8.5mm,inner sep=0pt,font=\normalsize},
  selected/.style={physical,draw=qbpteal!80!black,very thick,fill=qbpteal!14},
  frame/.style={draw=qbpnavy!35,rounded corners=3mm,fill=qbpnavy!2},
  label/.style={font=\normalsize\bfseries,text=qbpnavy},
  note/.style={font=\footnotesize,align=center,text=qbpgray}
]
  \path[frame] (-4.15,0.05) rectangle (4.15,3.05);
  \node[label] at (0,2.73) {(a) Logical interaction graph};

  \node[logical] (q0) at (-2.85,1.43) {\(q_0\)};
  \node[logical] (q1) at (-1.00,2.02) {\(q_1\)};
  \node[logical] (q2) at (-1.37,0.68) {\(q_2\)};
  \draw[very thick,qbpblue] (q0)--node[above,sloped,font=\footnotesize,fill=white,
    inner sep=1pt,text=black] {\(3\)} (q1);
  \draw[very thick,qbpteal!80!black] (q0)--node[below,sloped,font=\footnotesize,
    fill=white,inner sep=1pt,text=black] {\(1\)} (q2);
  \draw[very thick,qbporange] (q1)--node[right,font=\footnotesize,fill=white,
    inner sep=1pt,text=black] {\(1\)} (q2);
  \node[note,text=black] at (1.62,1.27)
    {edge labels \(=\Fb_{i,j}\)};

  \path[frame] (-4.15,-3.22) rectangle (4.15,-0.15);
  \node[label] at (0,-0.48) {(b) Physical four-cycle};

  \node[selected] (p0) at (-3.00,-1.25) {\(p_0\)};
  \node[physical] (p1) at (-1.62,-1.25) {\(p_1\)};
  \node[physical] (p2) at (-1.62,-2.25) {\(p_2\)};
  \node[physical] (p3) at (-3.00,-2.25) {\(p_3\)};
  \draw[very thick,qbpgray]
    (p0)--(p1)
    (p1)--(p2)
    (p2)--(p3)
    (p3)--(p0);
  \node[draw=qbpteal!80!black,rounded corners=2mm,fill=qbpteal!8,
    font=\footnotesize,align=center,inner sep=3pt] (allocation) at (1.58,-1.12)
    {partial allocation\\\(q_0\mapsto p_0\)};
  \draw[-{Latex[length=2mm]},thick,qbpteal!80!black]
    (allocation.west) .. controls (0.45,-0.70) and (-2.20,-0.68) ..
    (p0.north);
  \node[draw=qbporange,rounded corners=2mm,fill=qbporange!8,
    font=\footnotesize,align=center,inner sep=3pt] at (1.62,-2.18)
    {reflection fixes \(p_0,p_2\)\\and swaps \(p_1,p_3\)};
  \node[note] at (0,-2.96)
    {all four vertices are equivalent at the root};
\end{tikzpicture}
}
\caption{Logical edge labels are raw unordered-pair counts.
Toy allocation instance used to introduce the definitions.  The
four-cycle is deliberately symmetric.  At the root, all physical vertices are
equivalent.  Once \(q_0\) is placed on \(p_0\), the remaining reflection
makes \(p_1\) and \(p_3\) equivalent for the next branch.}
\label{fig:toy}
\end{figure}
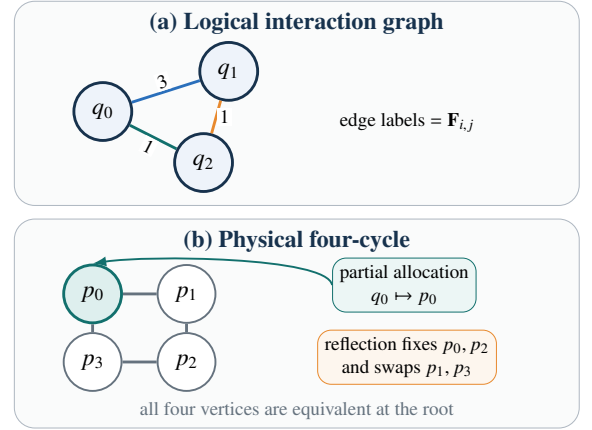

\subsection{Branch-and-bound state}

\begin{definition}[Partial allocation and completion]
\label{def:partial}
A \emph{partial allocation} maps a subset \({\cal A}\subseteq {\cal L}\) injectively to
physical device qubits. A \emph{completion} extends that map to an allocation of
all logical qubits to the physical ones.  At the corresponding search node,
\({\cal U} = {\cal L}\setminus {\cal A}\) denotes the unassigned logical qubits and
\({\cal R}={\cal P}\setminus\pi({\cal A})\) the free physical qubits.
\end{definition}

In our implementation, the code fixes an order \(i_0,\ldots,i_{n-1}\) of the logical
qubits. Following the traditional branch-and-bound process, a node at depth \(d\) has assigned
\({\cal A}=\{i_0,\ldots,i_{d-1}\}\), and its children place \(i_d\) on each
eligible vertex of \({\cal R}\).  The \emph{incumbent} \(K\) is the cost of the best
complete allocation known so far.  A child is discarded when its lower bound
cannot improve \(K\). The objective of any completion splits into three disjoint edge classes:
\begin{equation}
 Q(\pi)=C_{{\cal A},{\cal A}}(\pi)+C_{{\cal A},{\cal U}}(\pi)+C_{{\cal U},{\cal U}}(\pi).
 \label{eq:objective-partition}
\end{equation}
Only the first term is completely known at a partial node:
\begin{equation}
 C_{{\cal A},{\cal A}}(\pi)=\sum_{\substack{i<j\\i,j\in {\cal A}}}
 w_{i,j}\Db_{\pi(i),\pi(j)}.
 \label{eq:assigned-cost}
\end{equation}

\subsection{Feasible incumbent heuristics}
\label{subsec:incumbent-heuristics}

A lower bound proves that a branch cannot improve the current solution, whereas a feasible complete allocation supplies an \emph{upper bound}.  A
good upper bound can remove many branches before they are explored.  In this work, we use the following three deliberately simple heuristics to generate upper bounds on the exact objective in~\eqref{eq:objective}.

\paragraph{\underline{Multi-start greedy construction}}
The first heuristic fixes a logical priority
\(\rho=(i_0,\ldots,i_{n-1})\), tries every physical vertex as the image of
\(i_0\), and greedily completes each start.  At a partial greedy
mapping, the released code scores a free physical vertex \(p\) with the
one-sided matrix entries
\[
 \Delta_\Fb(i_d,p)=
 \sum_{t<d}\Fb_{i_t,i_d}\Db_{\pi(i_t),p}.
\]
The priority repeatedly removes the smallest-index logical vertex of minimum residual interaction weight and reverses the removal order. The released solver~\citep{valois2026}\footnote{\url{https://github.com/Guillaume-Helbecque/P3D-DFS}} uses the minimum cost over the \(N\) starts as its heuristic upper bound.

\begin{algorithm}[htbp]
\small
\caption{{\sc GreedyAllocation}}
\label{alg:released-greedy}
\begin{algorithmic}[1]
\Require {interaction matrix \(\Fb\), distances \(\Db\), logical
  priority \(\rho=(i_0,\ldots,i_{n-1})\)}
\State \(\pi^\star\gets\varnothing,\quad K\gets\infty\)
\ForAll{\(p_0\in P\)}
  \State start \(\pi\) with \(\pi(i_0)\gets p_0\)
  \For{\(d=1,\ldots,n-1\)}
    \State {\(p\gets
      \arg\min_{q\in P\setminus\pi(\operatorname{dom}\pi)}
      \sum_{t<d}F_{i_t,i_d}\Db_{\pi(i_t),q}\)}
    \State \(\pi(i_d)\gets p\)
  \EndFor
  \If{\(Q(\pi)<K\)}
    \State \((\pi^\star,K)\gets(\pi,Q(\pi))\)
  \EndIf
\EndFor
\State \Return \((\pi^\star,K)\)
\end{algorithmic}
\end{algorithm}

\paragraph{\underline{One-/two-exchange descent}}
Starting from a complete injection, the second heuristic considers two
neighbourhoods. A \emph{one-exchange} relocates one logical qubit to an
unused physical vertex; a \emph{two-exchange} swaps the physical images of
two logical qubits. 
The objective change of either move is computed from
only the incident logical edges.  For a move \(m\) applied to \(\pi\), let
\(\Delta Q(m)=Q(m(\pi))-Q(\pi)\), so a negative value is improving.  We
repeatedly take the best strictly
improving move. Because the set of injections is finite and every accepted
move decreases \(Q\), the procedure terminates at a local optimum.

\begin{algorithm}[htbp]
\small
\caption{{\sc{OneTwoDescent}}}
\label{alg:one-two-descent}
\begin{algorithmic}[1]
\Require a mapping \(\pi\), interaction weights \(w\), distances \(\Db\)
\Repeat
  \State \(m^\star\gets\varnothing,\quad\delta^\star\gets0\)
  \ForAll{valid relocations \(m:i\mapsto p\),
    \(p\in P\setminus\pi(L)\)}
    \If{\(\Delta Q(m)<\delta^\star\)}
      \State \((m^\star,\delta^\star)\gets(m,\Delta Q(m))\)
    \EndIf
  \EndFor
  \ForAll{swaps \(m:i\leftrightarrow j\)}
    \If{\(\Delta Q(m)<\delta^\star\)}
      \State \((m^\star,\delta^\star)\gets(m,\Delta Q(m))\)
    \EndIf
  \EndFor
  \If{\(m^\star\ne\varnothing\)}
    \State apply \(m^\star\) to \(\pi\)
  \EndIf
\Until{\(m^\star=\varnothing\)}
\State \Return \((\pi,Q(\pi))\)
\end{algorithmic}
\end{algorithm}

\paragraph{\underline{Budgeted randomised and iterated local search}}
This heuristic uses a time budget \(\tau\), including the
released greedy construction and its initial local descent.  It then
alternates between two starting mechanisms.  A
\emph{randomised greedy start} slightly perturbs the weighted-degree logical
order, chooses a random first physical vertex, and at every later step
chooses uniformly from a \emph{restricted candidate list} (RCL) containing
up to the three cheapest greedy candidates.  An
\emph{iterated-local-search start} applies between 2--5 random
relocations or swaps to the incumbent.  Both starts are followed by
Algorithm~\ref{alg:one-two-descent}; a descent started before the deadline is
allowed to finish.  The initial descent accepts only improving moves, so its
incumbent is no worse than the released greedy solution; every accepted
update also strictly improves it.  The restricted-candidate-list construction is a GRASP-style mechanism and is used here only to obtain feasible incumbents~\citep{feo1995}.

\begin{algorithm}[htbp]
\small
\caption{{\sc{BudgetedSearch}}}
\label{alg:budgeted-heuristic}
\begin{algorithmic}[1]
\Require {\(\Fb,\Db\), time budget \(\tau\), fixed seed \(s\)}
\State start the elapsed-time clock and seed the RNG with \(s\) 
\State {compute the released logical priority \(\rho\) from \(\Fb\)}
\State {\(w_{i,j}\gets F_{i,j}+F_{j,i}\) for every \(i<j\)}
\State {\((\pi_0,\cdot)\gets\Call{GreedyAllocation}{\Fb,\Db,\rho}\)}
\State \((\pi^\star,\cdot)\gets\Call{OneTwoDescent}{\pi_0,w,\Db}\)
\While{elapsed time \(<\tau\)}
  \If{the iteration number is even}
    \State {\((\pi,\cdot)\gets
      \Call{RandomisedGreedy}{\Fb,\Db,\mathrm{RCL}=3}\)}
  \Else
    \State \(\pi\gets\Call{Perturb}{\pi^\star,2\text{--}5\text{ moves}}\)
  \EndIf
  \State \((\pi,\cdot)\gets\Call{OneTwoDescent}{\pi,w,\Db}\)
  \If{\(Q(\pi)<Q(\pi^\star)\)}
    \State \(\pi^\star\gets\pi\)
  \EndIf
\EndWhile
\State \Return \((\pi^\star,Q(\pi^\star))\)
\end{algorithmic}
\end{algorithm}

\FloatBarrier
\subsection{{Hahn--Grant dual tensor bound (\HHB{})}}
\label{subsec:hhb-background}

Valois \emph{et al.} used this bound for their four-index implementation~\citep{hahn1998,hahn1998bnb,valois2026}\footnote{The four-index Hahn--Grant procedure used is different from the bound of Adams~\emph{et al.}~\citep{adams2007}.}.  At a partial allocation, \({\cal U}={\cal L}\setminus {\cal A}\) is the set of
unassigned logical qubits and \({\cal R}={\cal P}\setminus\pi({\cal A})\) is the set of free
physical qubits. 
Since the original circuit and device have \(n\) and \(N\) vertices,
respectively, \(|{\cal R}|-|{\cal U}|=N-n\geq0\). Let \({\cal Z}\) be a set of
\(|{\cal R}|-|{\cal U}|\) \emph{dummy logical vertices} with no interactions. Let $m$ be the order of the augmented logical set
\(\widetilde {\cal U}={\cal U}\cup {\cal Z}\) making the assignment problem square.

For \(i,k\in\widetilde {\cal U}\) and \(p,q\in {\cal R}\), two tentative assignments
\((i,p)\) and \((k,q)\) are compatible when \(i\neq k\) and \(p\neq q\). \HHB{} retains the fourth-order \emph{tensor}
\[
 \Tb_{i,p,k,q}=\Fb_{i,k}\Db_{p,q}
\]
for every compatible pair, together with an \(m\times m\)
\emph{leader matrix} \(\mathbf L\).  The entry \(\mathbf L_{i,p}\) stores the
current linear cost associated with the tentative assignment \(i\mapsto p\). An \HHB{} iteration (1) {\em distributes} every current leader value over its compatible tensor entries, (2) {\em equalises} complementary entries \(\Tb_{i,p,k,q}\) and \(\Tb_{k,q,i,p}\), (3) solves one Hungarian assignment in each of the \(m^2\) tensor blocks and moves the extracted block minima into \(\mathbf L_{i,p}\); and (4) solves a final Hungarian assignment on \(\mathbf L\) and adds its value to the node lower bound. Additional pruning mechanisms have also been used by Valois \emph{et al.}~\cite{valois2026}.

Each  \HHB{} iteration solves \(m^2\) assignments of order \(m\).  Consequently a full evaluation costs \(\mathcal{O}(I\,m^5)\) for $I$ iterations, stores \(\mathcal{O}(m^4)\) entries, and also pays \(\mathcal{O}(m^4)\) to create a child tensor.  

\subsection{Gilmore--Lawler assignment bound}
\label{subsec:glb-background}

The \emph{Gilmore--Lawler bound} (\GLB{}) replaces the residual quadratic problem by a rectangular linear assignment
\citep{gilmore1962,lawler1963}~(see also Anstreicher~\citep{anstreicher2003}). This is a rectangular \QAP{} setting because fewer logical objects than physical locations may be present~\citep{kaibel1998}. Following the previous subsection,
\({\cal U}\) and \({\cal R}\) are the unassigned logical and free physical sets,
\(u=|{\cal U}|\), and \(r=|{\cal R}|\), with \(u\le r\).  Unlike \HHB{}, \GLB{} does not
add dummy rows; it solves a rectangular \(u\times r\) assignment.  For a
tentative placement \(i\mapsto p\), the interactions from \(i\) to the assigned prefix have the exact cost
\begin{equation}
 \Ab_{i,p}=
 \sum_{a\in {\cal A}}
 \bigl(\Fb_{i,a}\Db_{p,\pi(a)}+\Fb_{a,i}\Db_{\pi(a),p}\bigr).
 \label{eq:glb-au}
\end{equation}
To relax the still-quadratic interactions within \({\cal U}\), it sorts the directed
weights \(\{\Fb_{i,j}:j\in {\cal U}\setminus\{i\}\}\) in nonincreasing order,
giving \(f_i^\downarrow(1),\ldots,f_i^\downarrow(u-1)\), then sorts the
distances \(\{\Db_{p,q}:q\in {\cal R}\setminus\{p\}\}\) in nondecreasing order and
retain the first \(u-1\), giving
\(d_{p,{\cal R}}^\uparrow(1),\ldots,d_{p,{\cal {\cal R}}}^\uparrow(u-1)\).  The rearrangement inequality gives the row cost
\begin{equation}
 \Gb_{i,p}=\Ab_{i,p}+
 \sum_{t=1}^{u-1}f_i^\downarrow(t)d_{p,{\cal R}}^\uparrow(t).
 \label{eq:glb-entry}
\end{equation}
The \GLB{} value is then
\begin{equation}
\small
 B_{\Gb}(\pi|_{\cal A})=C_{{\cal A},{\cal A}}(\pi)+
 \min_{\sigma:{\cal U}\rightarrow {\cal R}}\sum_{i\in {\cal U}}\Gb_{i,\sigma(i)}
 =C_{{\cal A},{\cal A}}(\pi)+\lap(\mathbf G).
 \label{eq:glb}
\end{equation}
Here \(\sigma\) ranges over injections from \({\cal U}\) into \({\cal R}\).  The
minimisation is an exact rectangular assignment and is solved by the
Hungarian algorithm \citep{kuhn1955}. Algorithms specialised to rectangular assignments are discussed in~\citep{bijsterbosch2010}. The quadratic problem has not become
linear: the construction has relaxed the requirement that the row-wise
distance pairings in \eqref{eq:glb-entry} arise from one common placement.

The validity of \eqref{eq:glb} follows from the standard
Gilmore--Lawler argument. For a fixed tentative assignment
\(i\mapsto p\), any completion assigns the other \(u-1\) rows
to \(u-1\) distinct columns of \({\cal R}\setminus\{p\}\). Their sorted
distances are componentwise no smaller than the \(u-1\) smallest
available distances retained in \eqref{eq:glb-entry}. Since the
interaction weights and distances are nonnegative, the
rearrangement inequality pairs descending weights with ascending
distances to obtain a lower bound on the directed interactions
from row \(i\). Summing over \(i\in {\cal U}\) counts every directed
term \(\Fb_{i,j}\) once, and the two directions together give
\(w_{i,j}=\Fb_{i,j}+\Fb_{j,i}\). Hence,
\[
  B_G(\pi|_{\cal A})\leq Q(\widehat{\pi})
\]
for every completion \(\widehat{\pi}\) of the partial allocation.

At one node, matrix assembly takes at most
cubic time in the residual order, and the rectangular Hungarian solve costs \(\mathcal{O}(u^2r)\). Thus the deployed bound is \(\mathcal{O}(N^3)\) in the worst case and uses \(\mathcal{O}(N^2)\) working memory.  This relatively small certificate is why the later conditioned parent-\LAP{} screen can reuse its matching and dual potentials.

\section{Graph-aware bounds and exact reductions}
\label{sec:bounds}

\subsection{Assigned cost filter: $P$}
The \emph{assigned-cost filter} \(P\) tests \eqref{eq:assigned-cost} before a
full residual bound is assembled.  It is especially cheap because assigning
one new logical qubit changes only the edges from that qubit to the existing
prefix.  If \(i\mapsto p\) is the next candidate, the update is
\begin{equation}
 C_{{\cal A}\cup\{i\},{\cal A}\cup\{i\}}=
 C_{{\cal A},{\cal A}}+\sum_{a\in A}w_{a,i}\Db_{\pi(a),p}.
 \label{eq:p-update}
\end{equation}
This costs \(\mathcal{O}(|{\cal A}|)\), after which the bound is available as a comparison.  If \(P\) rejects a candidate, no \GLB{} matrix or
\HHB{} child tensor is constructed. For every completion \(\widehat\pi\) of a node, all terms in \(C_{{\cal A},{\cal A}}\) are already fixed and nonnegative; hence
\[ Q(\widehat\pi)\geq C_{{\cal A},{\cal A}}(\pi). \]
This filter does not estimate an unknown cost; it simply records a cost that has already become unavoidable and will be taken into account soon by the \GLB{}. It therefore provides early work avoidance rather than a stronger retained-tree bound. Its runtime effect is small in our experiments. However, since the test is inexpensive, we retain it in the final configuration. In the toy example above, \(P\) obtains the lower bound \(6\) without solving an assignment problem. 

\subsection{Root and prefix symmetry:
\texorpdfstring{\(S_0\) and \(S_\ast\)}{S0 and S*}}

Physical-location symmetry in \QAP{} branch-and-bound is studied in the literature. Mautor and Roucairol define two free sites at a partial node as equivalent when a distance-preserving permutation maps one to the other while fixing all assigned sites pointwise, and Zhu \emph{et al.} apply this structural-symmetry reduction to the remaining physical qubits in exact qubit allocation~\citep{mautor1994,zhu2020}. 
Valois et~al. stated that the symmetry test of Zhu \emph{et al.} is expensive, as it is performed at every branch-and-bound tree node, and the benefits are highly restricted to the circuit~\citep{valois2026}. Here, our implementation difference is that the device automorphisms are generated and validated once, the active subgroup is carried as node state, and further group work is bypassed after only the identity remains.

An \emph{automorphism} of the device graph is a permutation
\(\gamma:{\cal P}\to {\cal P}\) that preserves all distances:
\(\Db_{\gamma(p),\gamma(r)}=\Db_{p,r}\).  The automorphisms form the group
\(\Gamma=\Aut(\Db)\).  To define automorphisms, every stored permutation is checked against
all \(N^2\) entries of \(\Db\).  For \(p\in {\cal P}\), its root \emph{orbit} is
\([p]_\Gamma=\{\gamma(p):\gamma\in\Gamma\}\).  Placing the first logical
qubit on any two vertices in the same orbit creates isomorphic subproblems.
The reduction \(S_0\), therefore, keeps the first vertex of each orbit in the
solver's fixed physical priority order.

After a prefix has been assigned, many root automorphisms no longer preserve that particular search node. At first, this may suggest that physical symmetry is useful only for the first allocation at the root. The prefix-dependent implementation \(S_\ast\) retains exactly the automorphisms that fix every occupied physical vertex. For a
partial allocation \(\pi|_{\cal A}\), this \emph{pointwise prefix stabilizer} is
\begin{equation}
 \Gamma_{\cal A}=\{\gamma\in\Aut(\Db):
 \gamma(\pi(a))=\pi(a)\text{ for every }a\in {\cal A}\}.
 \label{eq:stabilizer}
\end{equation}
The free candidates are partitioned into the dynamic orbits
\begin{equation}
 [p]_A=\{\gamma(p):\gamma\in\Gamma_{\cal A}\},\qquad p\in {\cal R},
 \label{eq:prefix-orbit}
\end{equation}
where \({\cal R}={\cal P}\setminus\pi({\cal A})\) is the set of currently free physical vertices. To prune equivalent branches, only the first physical-priority representative of each orbit is
branched on.  If the child chooses \(i\mapsto p\), its group is obtained by
the simple recurrence
\begin{equation}
 \Gamma_{{\cal A}\cup\{i\}}=\{\gamma\in\Gamma_{\cal A}:\gamma(p)=p\}.
 \label{eq:stabilizer-update}
\end{equation}
For completeness, the following proposition restates the
standard cost-preserving symmetry argument in the notation of the
present rectangular allocation problem.

\begin{proposition}[Safety of symmetry representatives]
\label{prop:symmetry}
Root-orbit pruning \(S_0\) and prefix-stabilizer pruning \(S_\ast\) preserve
the optimum allocation value.
\end{proposition}
\begin{proof}
Fix a partial allocation \(\pi|_A\), let \(i\) be the next logical qubit,
and consider two free physical vertices \(p,p'\in {\cal R}\) in the same
\(\Gamma_{\cal A}\)-orbit.  Suppose that \(p\) is the retained representative and
that the branch \(i\mapsto p'\) is discarded.  By the definition of an
orbit, there is a \(\gamma\in\Gamma_{\cal A}\) such that
\(\gamma(p')=p\).

Let \(\widehat\pi\) be any complete allocation extending the discarded
child, and define a new allocation by
\[
 \widehat\pi^\gamma(x):=\gamma(\widehat\pi(x))
 \qquad (x\in {\cal L}).
\]
Because \(\gamma\) is a permutation of the physical vertices,
\(\widehat\pi^\gamma\) remains injective.  Moreover, for every \(a\in {\cal A}\),
\[
 \widehat\pi^\gamma(a)
 =\gamma(\pi(a))=\pi(a),
\]
since \(\gamma\in\Gamma_{\cal A}\) fixes the occupied physical vertices
pointwise.  Thus \(\widehat\pi^\gamma\) extends the same prefix, and
\(\widehat\pi^\gamma(i)=\gamma(p')=p\), so it belongs to the retained
child. Finally, \(\gamma\) preserves every entry of the device-distance matrix.
Consequently,
\begin{align*}
 Q(\widehat\pi^\gamma)
 &=\sum_{x<y}w_{x,y}
   \Db_{\gamma(\widehat\pi(x)),\gamma(\widehat\pi(y))}
 &=\sum_{x<y}w_{x,y}\Db_{\widehat\pi(x),\widehat\pi(y)}
 =Q(\widehat\pi).
\end{align*}
The inverse automorphism \(\gamma^{-1}\in\Gamma_{\cal A}\) gives the reverse map.
Hence, the discarded and retained children have cost-preserving bijective
sets of completions and, therefore, the same minimum completion value.
Keeping one child per \(\Gamma_{\cal A}\)-orbit cannot change the minimum over all
children.  At the root \({\cal A}=\varnothing\), so
\(\Gamma_{\cal A}=\Gamma\), and the same argument proves the safety of \(S_0\).
\end{proof}

For the toy example in~\Cref{fig:dynamic-symmetry}, at the root, the four-cycle is
vertex-transitive, so \(S_0\) may branch only on \(p_0\).  After
\(q_0\mapsto p_0\), the reflection through \(p_0\) and \(p_2\) remains in
\(\Gamma_{\{q_0\}}\).  It exchanges \(p_1\) and \(p_3\), so \(S_\ast\) does not need
to explore both assignments. This illustrates the importance of dynamism: the available symmetry depends on the prefix, not merely on the device.

\begin{figure}[htbp]
\centering
\begin{tikzpicture}[
  >=Latex,
  card/.style={draw=qbpnavy!38,rounded corners=3mm,fill=white},
  free/.style={circle,draw=qbpgray,thick,fill=white,
    minimum size=8.5mm,inner sep=0pt,font=\small},
  allorbit/.style={free,draw=qbpblue,very thick,fill=qbpblue!12},
  occupied/.style={free,draw=qbpnavy,very thick,fill=qbpnavy,
    text=white},
  pair/.style={free,draw=qbpteal!80!black,very thick,fill=qbpteal!14},
  singleton/.style={free,draw=qbporange,very thick,fill=qbporange!13},
  edge/.style={very thick,qbpgray},
  note/.style={font=\footnotesize,align=center,text=qbpgray,
    text width=36mm},
  widenote/.style={note,text width=40mm},
  title/.style={font=\normalsize\bfseries,text=qbpnavy}
]
  \path[card,fill=qbpblue!2] (-4.05,0.15) rectangle (-0.10,4.25);
  \path[card,fill=qbpteal!2] (0.10,0.15) rectangle (4.05,4.25);
  \path[card,fill=qbporange!2] (-4.05,-2.85) rectangle (4.05,-0.05);
  \node[title] at (-2.08,3.90) {(a) Root: \(S_0\)};
  \node[title] at (2.08,3.90) {(b) Prefix: \(S_\ast\)};
  \node[title] at (0,-0.40) {(c) Identity reached};

  \node[allorbit] (a0) at (-2.78,2.95) {\(p_0\)};
  \node[allorbit] (a1) at (-1.38,2.95) {\(p_1\)};
  \node[allorbit] (a2) at (-1.38,1.70) {\(p_2\)};
  \node[allorbit] (a3) at (-2.78,1.70) {\(p_3\)};
  \draw[edge] (a0)--(a1);
  \draw[edge] (a1)--(a2);
  \draw[edge] (a2)--(a3);
  \draw[edge] (a3)--(a0);
  \node[note] at (-2.08,0.72)
    {\(\Gamma=\Aut(C_4)\)\\one orbit \(\{p_0,p_1,p_2,p_3\}\)\\keep \(p_0\)};

  \node[occupied] (b0) at (1.38,2.95) {\(p_0\)};
  \node[pair] (b1) at (2.78,2.95) {\(p_1\)};
  \node[singleton] (b2) at (2.78,1.70) {\(p_2\)};
  \node[pair] (b3) at (1.38,1.70) {\(p_3\)};
  \draw[edge] (b0)--(b1);
  \draw[edge] (b1)--(b2);
  \draw[edge] (b2)--(b3);
  \draw[edge] (b3)--(b0);
  \draw[dashed,very thick,qbpteal] (b0)--(b2);
  \node[note] at (2.08,0.72)
    {\(q_0\mapsto p_0\)\\orbits \(\{p_1,p_3\}\), \(\{p_2\}\)\\keep \(p_1,p_2\)};

  \node[occupied] (c0) at (-2.85,-1.15) {\(p_0\)};
  \node[occupied,fill=qbpteal!80!black,draw=qbpteal!80!black]
    (c1) at (-1.45,-1.15) {\(p_1\)};
  \node[singleton] (c2) at (-1.45,-2.35) {\(p_2\)};
  \node[allorbit] (c3) at (-2.85,-2.35) {\(p_3\)};
  \draw[edge] (c0)--(c1);
  \draw[edge] (c1)--(c2);
  \draw[edge] (c2)--(c3);
  \draw[edge] (c3)--(c0);
  \node[widenote] at (1.55,-1.70)
    {\(q_0\mapsto p_0,\ q_1\mapsto p_1\)\\
     \(\Gamma_{\cal A}=\{\mathrm{id}\}\)\\all free orbits are singletons};
\end{tikzpicture}
\caption{Dynamic physical symmetry on the toy four-cycle.  Root symmetry
keeps one representative of the full orbit.  Fixing \(p_0\) leaves a
reflection that exchanges \(p_1\) and \(p_3\).  Fixing \(p_1\) as well leaves
only the identity, after which \(S_\ast\) bypasses further group work.}
\label{fig:dynamic-symmetry}
\end{figure}
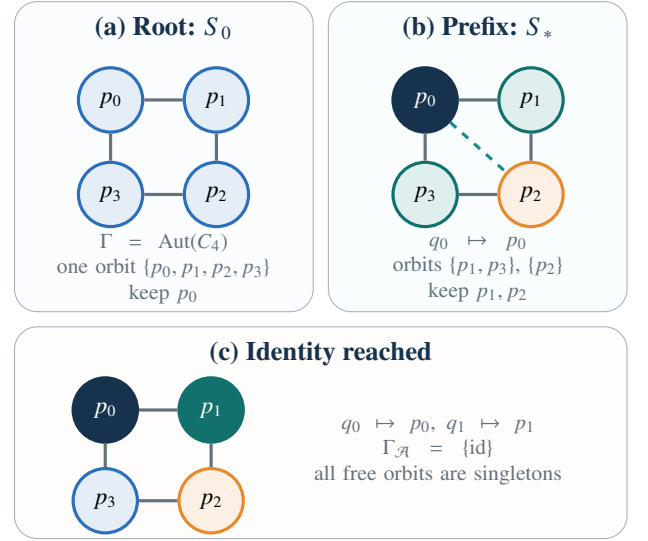

\subsection{{Conditioned parent-\LAP{} screening: \(L\)}}

The conditioned parent-\LAP{} screen \(L\) uses the parent's residual assignment certificate to reject candidates before a child matrix is constructed. It computes the exact optimum of the parent \LAP{} relaxation conditioned on the next assignment; it is not stronger than the \GLB{} rebuilt at that child.  Let the parent \GLB{} matrix be
\(\mathbf G\), and let its Hungarian solve return:
\begin{itemize}[noitemsep]
  \item an optimal row-to-column injection \(M\) of value
  \(\lambda=\lap(\mathbf G)\);
  \item feasible row and column dual potentials \(\alpha_i,\beta_p\).
\end{itemize}
We use the standard rectangular-\LAP{} dual
\[
 \max\left\{\sum_{i\in {\cal U}}\alpha_i+\sum_{p\in {\cal R}}\beta_p:
 \alpha_i+\beta_p\leq \Gb_{i,p},\ \beta_p\leq0\right\}.
\]
At an optimal primal--dual pair, every matched edge is tight and every column
unused by \(M\) has \(\beta_p=0\). The \emph{reduced cost}
\begin{equation}
 \overline \Gb_{i,p}=\Gb_{i,p}-\alpha_i-\beta_p
 \label{eq:reduced-cost}
\end{equation}
is nonnegative and is zero on every matched edge.  

Let \({\cal R}_M=M({\cal U})\) be the columns used by the optimal injection, and let
\({\cal R}_0={\cal R}\setminus {\cal R}_M\). The directed column residual graph
\(W(M)\) has a vertex set \({\cal R}\), together with one vertex \(\xi\)
when \({\cal R}_0\neq\varnothing\). Its edges are defined as follows.
\begin{itemize}[noitemsep]
  \item For every matched column \(q=M(k)\in {\cal R}_M\) and every \(s\in {\cal R}\),
  add the edge \(q\to s\) of weight \(\overline\Gb_{k,s}\). It represents
  reassigning row \(k\) from \(q\) to \(s\).
  \item When \(R_0\neq\varnothing\), add \(q\to\xi\) with weight zero for
  every \(q\in {\cal R}_0\), and add \(\xi\to s\) with weight \(-\beta_s\) for
  every \(s\in {\cal R}\). These edges contract the dummy-row transitions that
  change which physical column is unused.
\end{itemize}
All edge weights are nonnegative. Forcing row \(i\) to column \(p\) contributes
the reduced cost \(\overline\Gb_{i,p}\) and opens an alternating reassignment
path from \(p\) to the old column \(M(i)\). The exact increase within the
parent \LAP{} relaxation is therefore
\begin{equation}
 \rho_{i,p}=\overline \Gb_{i,p}+
 \operatorname{dist}_{W(M)}(p,M(i)),
 \qquad
 \lambda_{i,p}=\lambda+\rho_{i,p}.
 \label{eq:forced-lap}
\end{equation}
The dummy vertex represents changing which surplus column is unused.  One
shortest-path pass on the reversed residual graph prices every physical
candidate for the next logical row; solving a fresh Hungarian problem for
each child is unnecessary. This use of alternating paths and reduced costs follows classical assignment-sensitivity and ranked-matching principles~\citep{chegireddy1987}. The three subfigures of \Cref{fig:forced-lap} follow this calculation: (a) shows the parent matching and forced edge,
(b) shows the alternating return that repairs the matching, and
(c) converts its penalty into the conditioned parent-\LAP{} bound and screening decision.

\begin{figure}[htbp]
\centering
\begin{tikzpicture}[
  >=Latex,
  card/.style={draw=qbpnavy!38,rounded corners=2.5mm,fill=white},
  rnode/.style={circle,draw=qbpnavy,thick,fill=qbpblue!7,
    minimum size=6.8mm,inner sep=0pt,font=\scriptsize},
  cnode/.style={circle,draw=qbpteal!80!black,thick,fill=qbpteal!8,
    minimum size=6.8mm,inner sep=0pt,font=\scriptsize},
  match/.style={very thick,qbpteal!80!black},
  force/.style={-{Latex[length=1.8mm]},very thick,qbporange},
  residual/.style={-{Latex[length=1.8mm]},thick,qbpblue},
  box/.style={draw=qbpnavy,rounded corners=2mm,fill=white,
    minimum width=24mm,minimum height=14mm,inner sep=1.5pt,
    align=center,font=\scriptsize},
  note/.style={font=\scriptsize,align=center,text=qbpgray,text width=35mm},
  title/.style={font=\footnotesize\bfseries,text=qbpnavy}
]
  \path[card,fill=qbpteal!2] (-4.05,0.15) rectangle (-0.10,4.55);
  \path[card,fill=qbpblue!2] (0.10,0.15) rectangle (4.05,4.55);
  \path[card,fill=qbporange!2] (-4.05,-2.25) rectangle (4.05,-0.05);
  \node[title] at (-2.08,4.22) {(a) Parent certificate};
  \node[title] at (2.08,4.22) {(b) Alternating return};
  \node[title] at (0,-0.37) {(c) Child decisions};

  \node[rnode] (r1) at (-3.48,3.25) {\(i\)};
  \node[rnode] (r2) at (-3.48,2.43) {\(k\)};
  \node[rnode] (r3) at (-3.48,1.61) {\(\ell\)};
  \node[cnode] (p1) at (-0.67,3.55) {\(M(i)\)};
  \node[cnode] (p2) at (-0.67,2.75) {\(p\)};
  \node[cnode] (p3) at (-0.67,1.95) {\(q\)};
  \node[cnode,double] (p4) at (-0.67,1.15) {\(\xi\)};
  \draw[match] (r1)--(p1);
  \draw[match] (r2)--(p2);
  \draw[match] (r3)--(p3);
  \draw[force] (r1)--node[above,sloped,font=\scriptsize,
    fill=qbpteal!2,inner sep=1pt] {force} (p2);
  \node[note] at (-2.08,0.52)
    {green edges form \(M\); \(\xi\) is an artificial vertex};

  \node[cnode,draw=qbporange,fill=qbporange!10] (fp) at (0.75,3.25)
    {\(p\)};
  \node[cnode] (fz) at (2.08,2.55) {\(q\)};
  \node[cnode] (fm) at (3.40,3.25) {\(M(i)\)};
  \node[cnode,double,draw=qbpgray] (fd) at (2.08,1.30) {\(\xi\)};
  \draw[residual] (fp) to[bend left=13]
    node[above,sloped,font=\scriptsize] {reassign} (fz);
  \draw[residual] (fz) to[bend left=13]
    node[above,sloped,font=\scriptsize] {return} (fm);
  \draw[dashed,thick,qbpgray,-{Latex[length=1.8mm]}] (fp)--(fd);
  \draw[dashed,thick,qbpgray,-{Latex[length=1.8mm]}] (fd)--(fm);
  \node[note] at (2.08,0.52)
    {the shortest \(p\leadsto M(i)\) path repairs the matching};

  \node[box,draw=qbporange] (rho) at (-2.68,-1.36)
    {\(\rho_{i,p}=\overline \Gb_{i,p}\)\\
     \({}+\operatorname{dist}_{W}(p,M(i))\)};
  \node[box,draw=qbpblue,fill=qbpblue!7] (child) at (0,-1.36)
    {\(B_{L}(i\mapsto p)=\)\\
     \(C_{{\cal A},{\cal A}}+\lambda+\rho_{i,p}\)};
  \node[box,draw=qbpred,fill=qbpred!6] (prune) at (2.68,-1.36)
    {discard when \\\(B_{L}(i\mapsto p)\)\\cannot improve \(K\)};
  \draw[-{Latex[length=1.8mm]},thick,qbpgray]
    (rho.east)--(child.west);
  \draw[-{Latex[length=1.8mm]},thick,qbpgray]
    (child.east)--(prune.west);
\end{tikzpicture}
\caption{Conditioned parent-\LAP{} screening: The parent matching and duals define a residual graph.
Forcing one edge requires a cheapest alternating return to the old matched
column; its cost is the penalty within the parent \LAP{} relaxation.}
\label{fig:forced-lap}
\end{figure}
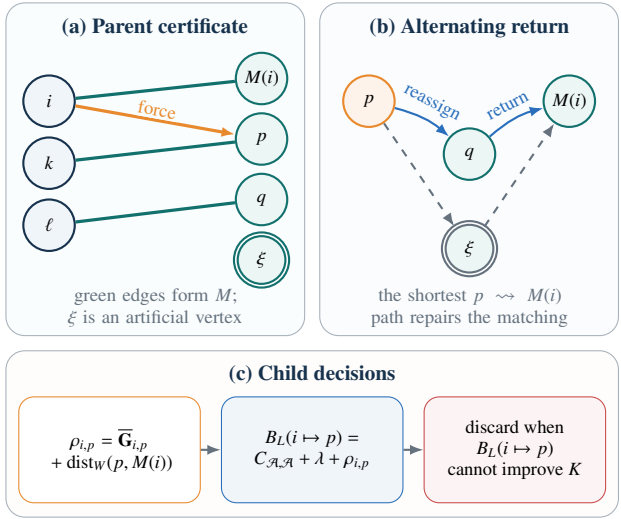

The exactness of \eqref{eq:forced-lap} follows from the standard
symmetric-difference argument for assignment matchings. Add \(r-u\)
zero-cost dummy rows so that \(M\) extends to a perfect matching. Replacing
\(i\mapsto M(i)\) by the forced edge \(i\mapsto p\) displaces the row
currently using \(p\), and feasibility is restored by an alternating
reassignment from \(p\) to \(M(i)\); transitions through \(\xi\) cover the
case in which the identity of an unused column changes. The cost of this
reassignment is exactly the sum of the corresponding reduced-cost edges.
Conversely, the symmetric difference between \(M\) and any assignment
\(\sigma\) satisfying \(\sigma(i)=p\) consists of this alternating component
and possibly additional alternating cycles. Since all residual-graph edge
weights are nonnegative, the additional cycles cannot improve the
conditioned assignment and may be discarded. Therefore,
\[
 \lambda_{i,p}
 =
 \min_{\substack{\sigma:{\cal U}\hookrightarrow {\cal R}\\ \sigma(i)=p}}
 \sum_{k\in {\cal U}}\Gb_{k,\sigma(k)}
 =
 \lambda+\overline\Gb_{i,p}
 +\operatorname{dist}_{W(M)}(p,M(i)).
\]
Consequently,
\[
  B_{L}(i,p)
  :=
  C_{{\cal A},{\cal A}}+\lambda_{i,p}
\]
is a valid lower bound for every completion of the child
\(i\mapsto p\): the conditioned parent LAP remains a relaxation
of that child's residual quadratic allocation problem.

\begin{proposition}[Dominance over assigned-cost screening]
\label{prop:l-dominates-p}
For every candidate \(i\mapsto p\), let
\[
 B_{P}(i,p)=C_{{\cal A},{\cal A}}+\Ab_{i,p}
\]
be the bound tested by \(P\).  Then
\[
 B_{P}(i,p)\leq B_{L}(i,p).
\]
Therefore, at any depth where \(L\) is active, a candidate that survives the
\(L\) test cannot subsequently be rejected by the \(P\) comparison under
either the strict or non-strict incumbent rule.
\end{proposition}
\begin{proof}
By \eqref{eq:glb-entry}, the forced row contributes
\(\Gb_{i,p}=\Ab_{i,p}+g_i^{{\cal U},{\cal R}}(p)\geq\Ab_{i,p}\), because interactions and
distances are nonnegative.  Every other entry selected by the forced assignment is also nonnegative.  Hence
\(\lambda_{i,p}\geq\Ab_{i,p}\); adding \(C_{{\cal A},{\cal A}}\) gives the result.
\end{proof}

Due to this dominance, we integrated it to the proposed solver as follows: The vector \(\{B_L(i,p):p\in {\cal R}\}\) is prepared once per parent, or carried by
\(\mathrm{CERT}\)~(of the engineering bundle described below) from the parent's child-bound computation. Its
per-candidate test is therefore one array access and one comparison.  Testing
it before \(P\) avoids the \(\mathcal O(|{\cal A}|)\) assigned-cost update for every
candidate rejected by \(L\).  For an \(L\) survivor, the implementation still
computes \(B_P\), because that exact assigned cost is stored in the child and
is the fixed-cost term of the child \GLB{}; only the subsequent \(P\) comparison is redundant.  

Although valid, this conditioned parent-LAP bound is also no stronger
than the GLB obtained by constructing and bounding the child.
Let
\[
 g_i^{{\cal U},{\cal R}}(p)=\sum_{t=1}^{|{\cal U}|-1}
 f_i^{{\cal U},\downarrow}(t)d_{p,{\cal R}}^{\uparrow}(t)
\]
be the row-relaxation term in \eqref{eq:glb-entry}, and let
\(B_{G}^{\mathrm{child}}(i,p)\) denote the \GLB{} obtained after actually
creating the child \(i\mapsto p\) and rebuilding its residual matrix.

\begin{proposition}[Dominance by the rebuilt child \GLB{}]
\label{prop:l-dominated}
For every child \(i\mapsto p\),
\[
 {B_L(i,p)\leq B_{G}^{\mathrm{child}}(i,p)}.
\]
Consequently, at a fixed cutoff and with the same incumbent sequence, \(L\)
cannot remove a node that the rebuilt child \GLB{} would retain; its benefit
is avoiding the construction and solution of child bounds that will be
rejected.
\end{proposition}
\begin{proof}
Let \({\cal U}'={\cal U}\setminus\{i\}\), \({\cal R}'={\cal R}\setminus\{p\}\), and fix any injection
\(\sigma:{\cal U}'\to {\cal R}'\). For the forced row, the columns \(\sigma({\cal U}')\) give one feasible selection and pairing in the parent row relaxation, so
\[
 g_i^{{\cal U},{\cal R}}(p)\leq
 \sum_{k\in {\cal U}'}\Fb_{i,k}\Db_{p,\sigma(k)}.
\]
For each remaining row \(k\in {\cal U}'\), pairing the weight \(\Fb_{k,i}\) with the fixed distance \(\Db_{\sigma(k),p}\) and using the child row relaxation for the other weights gives a feasible pairing for the parent row relaxation. Hence,
\[
 g_k^{{\cal U},{\cal R}}(\sigma(k))\leq
 \Fb_{k,i}\Db_{\sigma(k),p}
 +g_k^{{\cal U}',{\cal R}'}(\sigma(k)).
\]
Moreover, assigning \(i\mapsto p\) changes the fixed cost by
\(\Ab_{i,p}\), while each child assigned--unassigned entry adds
\(\Fb_{k,i}\Db_{\sigma(k),p}+\Fb_{i,k}\Db_{p,\sigma(k)}\).
Summing the displayed inequalities therefore shows, for every \(\sigma\),
that the conditioned parent-\LAP{} objective is no larger than the rebuilt
child-\GLB{} objective. Minimising both sides over \(\sigma\) proves the
claim.
\end{proof}

For a toy example, suppose the residual matrix has rows
\(q_1,q_2\), columns \(p_1,p_2,p_3\), and values
\[
 \mathbf{C}=\begin{pmatrix}0&5&0\\2&0&2\end{pmatrix}.
\]
Its unforced value is \(0\), attained by matching
\(q_1\) to \(p_1\) and \(q_2\) to \(p_2\).  If the next branch forces
\(q_1\mapsto p_2\), the best remaining choice costs \(5+2=7\). Thus, the
single parent certificate identifies a conditioned parent-\LAP{} penalty of \(7\) for
that child.  
In our preliminary experiments, we observed that $L$ is useful when the extra certificate work is amortised over
many candidate children. Hence, in the solver, we apply
it only through depth seven.

\subsection{The engineering bundle: \texorpdfstring{\(E\)}{E}}

The final configuration contains an \emph{engineering bundle}, denoted
\(E\), in addition to \(P\), \(S_{0/\ast}\), and \(L\).  In this work, we deliberately separate it from the combinatorial bundle, because it preserves the same lower bounds and child decisions for a fixed traversal order and incumbent sequence. That is, it does not provide additional pruning. This bundle has five parts implemented on top of the combinatorial bundle:
\[
 E=\mathrm{INC}+\mathrm{DDOT}+\mathrm{PAU}+
   \mathrm{CERT}+\mathrm{TINY}.
\]
%They are described separately because their setup costs and reuse scopes are different.

\subsubsection{Incremental node state
  \texorpdfstring{(\(\mathrm{INC}\))}{(INC)}}

Every node carries the already assigned cost \(C_{{\cal A},{\cal A}}\), the free-set bit
mask \(M_{\cal R}\), and the active automorphism mask \(M_{\cal A}\).  For a child
\(i\mapsto p\), \eqref{eq:p-update} updates the first value,
\(M_{\cal R}\mathrel{\&}=\mathord\sim(1\ll p)\) updates the second, and the
intersection in \eqref{eq:stabilizer-update} updates the third.  This removes
repeated scans of the full partial mapping.  The stored values are exact
state, not lower-bound approximations.

\subsubsection{Reusable device profiles and exact dots
  \texorpdfstring{(\(\mathrm{DDOT}\))}{(DDOT)}}

The expensive part of \eqref{eq:glb-entry} repeatedly asks the same
device-only question: for a free set \({\cal R}\) and candidate \(p\), how many
vertices occur at each routing distance?  Define the \emph{distance profile}
\begin{equation}
 \nu_{{\cal R},p}(\delta)=
 \bigl|\{q\in {\cal R}\setminus\{p\}:\Db_{p,q}=\delta\}\bigr|.
 \label{eq:distance-profile}
\end{equation}
This histogram determines the sorted vector
\(d^\uparrow_{p,{\cal R}}\) without sorting vertex identifiers.  
When the memory budget permits, this is an intentionally heavy, exhaustive device-compilation stage.  Every integer mask \(m\in\{0,\ldots,2^N-1\}\) represents one possible
free physical set.  For every set bit \(p\) in \(m\), the preprocessor evaluates
\(\nu_{{\cal R},p}\) and replaces identical histograms of the same cardinality by
one compact profile identifier.

\begin{figure*}[htbp]
\centering
\begin{tikzpicture}[
  >=Latex,
  card/.style={draw=qbpnavy!38,rounded corners=3mm,fill=white},
  free/.style={circle,draw=qbpteal!80!black,very thick,fill=qbpteal!14,
    minimum size=6mm,inner sep=0pt},
  used/.style={circle,draw=qbpgray,thick,fill=qbpgray!14,
    minimum size=6mm,inner sep=0pt,font=\scriptsize},
  candidate/.style={circle,draw=qbpnavy,very thick,fill=qbpnavy,
    text=white,minimum size=7mm,inner sep=0pt,font=\scriptsize},
  tile/.style={draw=qbpblue,rounded corners=2mm,fill=qbpblue!8,
    minimum width=8mm,minimum height=8mm,font=\scriptsize},
  dtile/.style={tile,draw=qbporange,fill=qbporange!10},
  title/.style={font=\small\bfseries,text=qbpnavy},
  note/.style={font=\scriptsize,align=center,text=qbpgray,text width=42mm}
]
  \path[card,fill=qbpteal!2] (-7.35,-2.35) rectangle (-2.55,2.35);
  \path[card,fill=qbpblue!2] (-2.30,-2.35) rectangle (2.35,2.35);
  \path[card,fill=qbporange!2] (2.60,-2.35) rectangle (7.35,2.35);
  \node[title] at (-4.95,2.00) {(a) All-subset cache};
  \node[title] at (0.02,2.00) {(b) Per-circuit dots};
  \node[title] at (4.98,2.00) {(c) Node hot path};

  \begin{scope}[shift={(-4.95,0.30)}]
    \fill[qbpblue!7] (0,0) circle (1.36);
    \fill[qbporange!11] (0,0) circle (0.72);
    \draw[dashed,thick,qbpblue!65] (0,0) circle (1.36);
    \draw[dashed,thick,qbporange!75] (0,0) circle (0.72);
    \node[candidate] (cp) at (0,0) {\(p\)};
    \node[free] (f1) at (-0.55,0.22) {};
    \node[free] (f2) at (0.48,-0.35) {};
    \node[used] (x1) at (0.20,0.62) {\(\times\)};
    \node[free] (f3) at (-0.92,0.76) {};
    \node[used] (x2) at (1.05,-0.55) {\(\times\)};
    \draw[thick,qbpgray] (cp)--(f1) (cp)--(f2) (cp)--(x1);
    \draw[thick,qbpgray] (f1)--(f3) (f2)--(x2);
  \end{scope}
  \node[note] at (-4.95,-1.50)
    {\(\nu_{{\cal R},p}=(2,1,\ldots)\mapsto\rho\)\\
     \(\rho=I[o[M_R]+\operatorname{rank}_{M_{\cal R}}(p)]\)};
  \node[note] at (-4.95,-2.08)
    {all \(2^N\) masks; device only};

  \node[note,text=qbpnavy] at (0.02,1.45)
    {row \(i\): descending remaining interactions};
  \foreach \x/\v in {0/8,1/5,2/3} {
    \node[tile] (w\x) at (-0.90+0.90*\x,0.74) {\(\v\)};
  }
  \node[note] at (0.02,0.02)
    {\(f_i^\downarrow=(8,5,3)\)\\online, circuit-dependent data};
  \node[draw=qbpnavy!55,rounded corners=2mm,fill=white,
    minimum width=36mm,minimum height=9mm,align=center,font=\scriptsize]
    at (0.02,-1.02)
    {same logical suffix is reused\\across many free-set profiles};
  \node[note] at (0.02,-1.88)
    {all circuit work is included in reported time};

  \node[note,text=qbpnavy] at (4.98,1.55)
    {profile expands to sorted routing distances};
  \foreach \x/\v in {0/0,1/0,2/1} {
    \node[dtile] (d\x) at (3.93+1.05*\x,0.88) {\(\v\)};
  }
  \node[draw=qbpgreen!70!black,rounded corners=2mm,fill=qbpgreen!8,
    minimum width=38mm,minimum height=9mm,align=center,font=\footnotesize]
    at (4.98,-0.16)
    {\(\Delta_{d,i,\rho}=8\cdot0+5\cdot0+3\cdot1=3\)};
  \node[draw=qbporange,rounded corners=2mm,fill=white,
    minimum width=38mm,minimum height=9mm,align=center,font=\footnotesize]
    at (4.98,-1.37)
    {\(\Gb_{i,p}=\Ab_{i,p}+\Delta_{d,i,\rho}=\Ab_{i,p}+3\)};
  \node[note] at (4.98,-2.05)
    {identical to direct sorting and dotting};
\end{tikzpicture}
\caption{The heavy \(\mathrm{DDOT}\) pipeline.  The left subfigure exhaustively compiles every free physical mask once per device.  The middle subfigure performs the circuit-dependent work inside the measured run.  The right subfigure is the
exact indexed lookup executed while assembling each residual \GLB{} matrix.}
\label{fig:device-profile-dot}
\end{figure*}
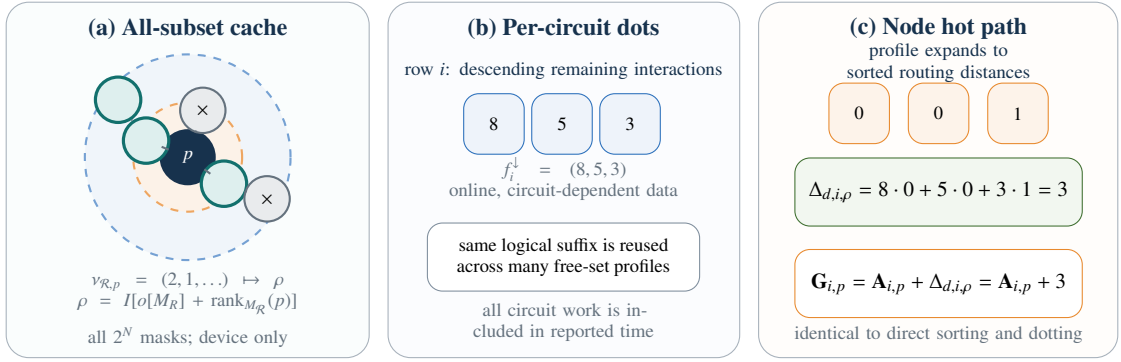

The persistent representation has an offset array \(o[0\ldots2^N]\) and a
packed identifier array \(I\).  If
\[
 \operatorname{rank}_{m}(p)=
 {\operatorname{popcount}}\bigl(m\mathbin{\&}((1\ll p)-1)\bigr),
\]
then the profile used by the solver is obtained directly as
\begin{equation}
 \rho(m,p)=I\!\left[o[m]+\operatorname{rank}_{m}(p)\right].
 \label{eq:profile-lookup}
\end{equation}
The preprocessor therefore stores \(2^N+1\) offsets and exactly
\[
 \sum_{m=0}^{2^N-1}\operatorname{popcount}(m)=N2^{N-1}
\]
mask--candidate identifiers.  Profile signatures store the unique
histograms.  Both arrays depend only on \(\Db\): their potentially expensive construction is paid once per physical device, and the resulting artifact is reused unchanged for every circuit. The sizes of artifacts for the devices used are given in~\Cref{tab:device-profile-cache}.

\begin{table}[htbp]
\centering
\footnotesize
\renewcommand{\arraystretch}{1.1}
\caption{Device profiles: The number of masks \(2^N\) is shown with the device name. ``IDs'' counts mask--candidate profile identifiers; ``Prof.'' counts unique profile signatures; ``Payload'' is the profile size.}
\label{tab:device-profile-cache}
\begin{tabular}{lrrr}
\toprule
Device (Masks) & IDs & Prof. & Payload\\
\midrule
{\tt Melbourne} (\(2^{16}\)) & \(524,288\) & \(6,404\) &
   \(\approx1.4\) MB\\
{\tt Boeblingen} \((2^{20}\)) & \(10,485,760\) & \(18,032\) &
   \(\approx25.3\) MB \\
{\tt Cairo} (\(2^{27}\)) & \(1,811,939,328\) & \(1,593,009\) &
  \(\approx7.8\) GB \\
\bottomrule
\end{tabular}
\end{table}

Architecture-level caching is not new in general; recent routing work caches distances, paths, and execution-subgraph information for a fixed architecture~\citep{russon2026}. Our main contribution here is the all-subset free-set/candidate profile representation and {\em its use in exact \GLB{} assembly for an exact solver}, rather than the general observation that a fixed architecture can be preprocessed.  

\paragraph{\underline{DDOT memory representation}}
Let \(T_k\) be the number of distinct
distance profiles for free sets of cardinality \(k\), let
\(T=\sum_{k=1}^{N}T_k\), and let \(b_{\Db}\) be the number of distance
buckets. A full DDOT artifact occupies
\begin{equation}
\begin{split}
 M_{\mathrm{DDOT}}={}&272+4(2^N+1)+4N2^{N-1}\\
 &+4(N+1)+4(N+2)
 +8\left\lceil\frac{b_{\Db}}{8}\right\rceil T
 \quad\text{bytes}.
\end{split}
\label{eq:ddot-memory}
\end{equation}
The terms are, respectively, the header, mask-offset array,
mask--candidate profile identifiers, per-cardinality profile counts,
signature offsets, and packed profile signatures.  Mask offsets, profile
identifiers, and profile counts are unsigned 32-bit integers.  A signature
uses 64-bit words, each containing eight unsigned 8-bit bucket counts.
For {\tt Cairo}, \(N=27\), \(b_{\Db}=12\), and \(T=1,593,009\). Its exact artifact size is \(7.81\)~GB, with the breakdown in~\Cref{tab:ddot-memory-breakdown}: Profile identifiers account for
\(92.80\%\) of the file, mask offsets for \(6.87\%\), and signatures for
\(0.33\%\).  

\begin{table}[htbp]
\centering
\scriptsize
\setlength{\tabcolsep}{2.5pt}
\caption{Exact memory composition of the Cairo DDOT artifact.}
\label{tab:ddot-memory-breakdown}
\resizebox{0.85\columnwidth}{!}{%
\begin{tabular}{@{}lrrr@{}}
\toprule
Component & Representation & Elements & Bytes\\
\midrule
Header             & byte/\texttt{uint64} & ---              & \(272\)\\
Mask offsets       & \texttt{uint32}       & \(134,217,729\)& \(536,870,916\)\\
Profile identifiers& \texttt{uint32}       & \(1,811,939,328\)& \(7,247,757,312\)\\
Profile counts     & \texttt{uint32}       & \(28\)           & \(112\)\\
Signature offsets  & \texttt{uint32}       & \(29\)           & \(116\)\\
Profile signatures & \(2\times\)\texttt{uint64} &
 \(1,593,009\) & \(25,488,144\)\\
\midrule
Total & & & \(7,810,116,872\)\\
\bottomrule
\end{tabular}}
\end{table}

Although exponential, the generation is a one-time, per-device cost whose artifact can be reused across circuits. To mitigate the scalability issues due to memory limits, a selective cache may instead materialise only chosen free-set cardinalities or encountered masks, trading some runtime benefit for lower storage. We leave this as future work.

\paragraph{\underline{DDOT mechanics}}
During a run, the solver still constructs the descending
logical suffix weights \(f_{d,i}^\downarrow\) at depth \(d\).  For each
remaining logical row and each profile identifier of the corresponding free cardinality it computes the exact dot product
\begin{equation}
 \Delta_{d,i,\rho}=
 \sum_{t=1}^{u-1}f_{d,i}^\downarrow(t)d_{\rho}^\uparrow(t).
 \label{eq:ddot}
\end{equation}
This circuit-dependent table is built inside the solver invocation.
At a search node, \(\mathrm{INC}\) already supplies its current binary free mask \(M_R\).  The hot path that assembles a \GLB{} entry is consequently
\begin{equation}
 \Gb_{i,p}=\Ab_{i,p}+
 \Delta_{d,i,\rho(M_R,p)}.
 \label{eq:cached-glb-entry}
\end{equation}
Thus, the exhaustive device table 
is consumed for every residual matrix assembled by every
(\(\mathrm{DDOT}\)-enabled) search. Equation~\eqref{eq:profile-lookup} selects
the device profile, and \eqref{eq:cached-glb-entry} loads the exact
row-relaxation term in \eqref{eq:glb-entry} in \(\mathcal{O}(1)\) time per matrix entry, replacing the
distance extraction and sort. \Cref{fig:device-profile-dot} summarizes this three-stage data flow and its
timing boundary.  Device-profile compilation in~(a) is performed once
offline, whereas the per-circuit dot materialisation in~(b) and the
per-node indexed lookup in~(c) are both included in measured solver
time.

\begin{figure*}[htbp]
\centering
\begin{tikzpicture}[
  x=1.05cm,y=1cm,
  device/.style={circle,draw=qbpnavy,thick,fill=white,
    minimum size=4.7mm,inner sep=0pt,font=\scriptsize},
  link/.style={thick,qbpgray},
  bridge/.style={very thick,qbpred,densely dashed},
  card/.style={draw=qbpnavy!28,rounded corners=2.5mm},
  title/.style={font=\small\bfseries,text=qbpnavy},
  badge/.style={draw=qbpnavy!35,rounded corners=1.5mm,fill=white,
    font=\scriptsize,text=qbpnavy,inner xsep=4pt,inner ysep=2pt}
]
  \filldraw[card,fill=qbpblue!3] (-7.45,0.20) rectangle (-0.30,3.32);
  \filldraw[card,fill=qbpteal!3] (0.30,0.20) rectangle (7.45,3.32);
  \filldraw[card,fill=qbporange!3] (-7.45,-3.75) rectangle (7.45,-0.14);

  % Melbourne: the exact 16-vertex ladder labelling used by the matrix.
  \node[title] at (-3.88,3.02) {(a) \texttt{Melbourne}};
  \foreach \col in {0,...,7}{
    \pgfmathtruncatemacro{\bottomid}{15-\col}
    \node[device] (mt\col) at (-6.55+0.77*\col,2.15) {\col};
    \node[device] (mb\col) at (-6.55+0.77*\col,1.23) {\bottomid};
  }
  \foreach \col [evaluate=\col as \next using int(\col+1)] in {0,...,6}{
    \draw[link] (mt\col)--(mt\next);
    \draw[link] (mb\col)--(mb\next);
  }
  \foreach \col in {0,...,7}{\draw[link] (mt\col)--(mb\col);}
  \node[badge] at (-3.88,0.52)
    {$N=16$, $|{\cal E}| = 22$, $|\Aut(G_D)|=4$};

  % Boeblingen: exact matrix graph; red links are its four bridges.
  \node[title] at (3.88,3.02) {(b) \texttt{Boeblingen}};
  \foreach \row in {0,...,3}{
    \foreach \col in {0,...,4}{
      \pgfmathtruncatemacro{\id}{5*\row+\col}
      \node[device] (b\id) at (1.52+1.18*\col,2.57-0.57*\row) {\id};
    }
  }
  \foreach \row in {0,...,3}{
    \foreach \col [evaluate=\col as \next using int(\col+1)] in {0,...,3}{
      \pgfmathtruncatemacro{\a}{5*\row+\col}
      \pgfmathtruncatemacro{\b}{5*\row+\next}
      \draw[link] (b\a)--(b\b);
    }
  }
  \foreach \a/\b in {1/6,3/8,5/10,7/12,9/14,11/16,13/18}{
    \draw[link] (b\a)--(b\b);
  }
  \foreach \a/\b in {0/1,3/4,15/16,18/19}{
    \draw[bridge] (b\a)--(b\b);
  }
  \node[badge] at (3.88,0.43)
    {$N=20$, $|{\cal E}| = 23$, $|\Aut(G_D)|=4$};

  % Cairo: three internally disjoint core paths and six pendant bridges.
  \node[title,anchor=west] at (-6.95,-0.46) {(c) \texttt{Cairo}};
  \node[badge,anchor=east] at (7.25,-0.46)
    {$N=27$, $|{\cal E}| = 28$, $|\Aut(G_D)|=2$};
  \node[device] (c12) at (-6.15,-2.08) {12};
  \node[device] (c14) at (6.15,-2.08) {14};
  \node[device] (c13) at (0,-2.08) {13};
  \draw[link] (c12)--(c13)--(c14);

  \foreach \id/\xx/\yy in {
    10/-4.92/-1.62,7/-3.69/-1.36,4/-2.46/-1.22,
    1/-1.23/-1.14,2/0/-1.11,3/1.23/-1.14,
    5/2.46/-1.22,8/3.69/-1.36,11/4.92/-1.62}{
    \node[device] (c\id) at (\xx,\yy) {\id};
  }
  \draw[link]
    (c12)--(c10)--(c7)--(c4)--(c1)--(c2)--(c3)--(c5)--(c8)--(c11)--(c14);

  \foreach \id/\xx/\yy in {
    15/-4.92/-2.54,18/-3.69/-2.80,21/-2.46/-2.94,
    23/-1.23/-3.02,24/0/-3.05,25/1.23/-3.02,
    22/2.46/-2.94,19/3.69/-2.80,16/4.92/-2.54}{
    \node[device] (c\id) at (\xx,\yy) {\id};
  }
  \draw[link]
    (c12)--(c15)--(c18)--(c21)--(c23)--(c24)--(c25)--(c22)--(c19)--(c16)--(c14);

  \foreach \leaf/\base/\xx/\yy in {
    0/1/-1.55/-0.54,6/7/-4.10/-0.70,9/8/3.30/-0.70,
    17/18/-3.94/-3.43,20/19/3.94/-3.43,26/25/1.74/-3.49}{
    \node[device,draw=qbpred] (c\leaf) at (\xx,\yy) {\leaf};
    \draw[bridge] (c\leaf)--(c\base);
  }
  \node[font=\scriptsize,text=qbpred,anchor=west] at (-6.95,-3.48)
    {red dashed link $=$ bridge};
\end{tikzpicture}
\caption{Physical coupling graphs used in the experiments, reconstructed
from the distance matrices consumed by the solver.  Vertices are physical
qubits, and links are permitted two-qubit interactions.  Dashed links mark bridges. Although the graphs do not fit the manufacturer's chip geometry, the labels and adjacencies are exact.}
\label{fig:benchmark-devices}
\end{figure*}
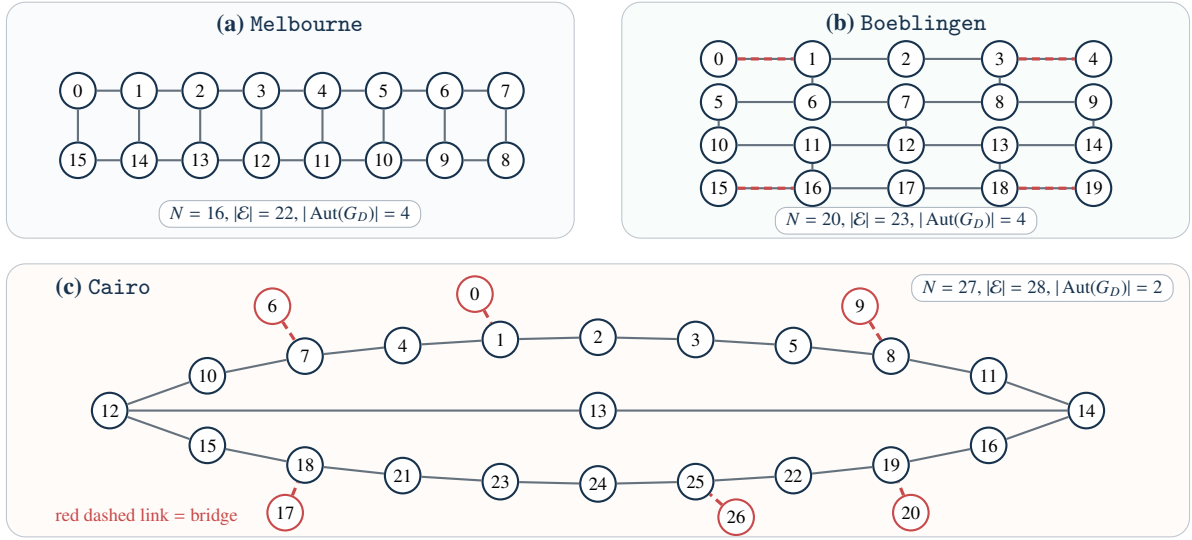

\subsubsection{Parent assigned--unassigned reuse
  \texorpdfstring{(\(\mathrm{PAU}\))}{(PAU)}}

At one expanded parent, define
\[
 \Ab^{\mathrm{prnt}}_{i,p}=
 \sum_{a\in A}
 \bigl(\Fb_{i,a}\Db_{p,\pi(a)}+\Fb_{a,i}\Db_{\pi(a),p}\bigr)
 \qquad(i\in {\cal U},\ p\in {\cal R}).
\]
All children share this table.  If the child adds \(i^\star\mapsto p^\star\),
the entry for a remaining row \(i\) and free column \(p\) is
\begin{equation}
 \Ab^{\mathrm{chld}}_{i,p}=\Ab^{\mathrm{prnt}}_{i,p}
 +\Fb_{i, i^\star}\Db_{p,p^\star}
 +\Fb_{i^\star, i}\Db_{p^\star,p}.
 \label{eq:pau}
\end{equation}
Thus the complete assigned prefix is scanned once per expanded parent, not once for every child matrix.

\subsubsection{Certificate reuse
  \texorpdfstring{(\(\mathrm{CERT}\))}{(CERT)}}

When a child's \GLB{} is evaluated, its Hungarian matching, duals, and forced values for the child's next logical row already exist.  If that child survives and is later expanded, \(\mathrm{CERT}\) carries those exact forced-\LAP{} values in the node.  Operator \(L\) consumes them directly instead of assembling and solving the same child assignment a second time. Rejected children carry no certificate.

\subsubsection{Tiny exact assignment kernels
  \texorpdfstring{(\(\mathrm{TINY}\))}{(TINY)}}

For the two remaining rows, the optimum is found in \(\mathcal{O}(r)\) time from the best and second-best columns of the second row.  For three rows, a scan over the columns maintains the eight subsets of \(\{1,2,3\}\), allowing each column to extend a state by at most one row.  Both kernels solve exactly the same rectangular assignment as Hungarian; exhaustive random and search-call verification compares them against the generic implementation.

\section{Experimental results}
\label{sec:experiments}
\paragraph{\underline{Runtime environment}} All experiments were performed on one four-socket Intel Xeon E7-4870 v2 server running at 2.30 GHz.  Each socket contains 15 physical cores, giving 60 cores in total.  The host has four NUMA domains and 503 GiB of memory and runs Ubuntu 22.04. The solver is implemented in Chapel 2.8.0 and compiled with \cfg{--fast}, the LLVM/Clang 14 backend, and Qthreads. All runs were executed with a 21,600-second timeout. The feasible-incumbent generator is the released standard-library Python
3.10.12 implementation; it prints a verified mapping and exact integer cost,
and that cost is passed unchanged to the Chapel solver.

\paragraph{\underline{Datasets}} All inputs come from the benchmark instances of~\citep{valois2026}. There are three devices: {\tt Melbourne} ($N = 16$, 21 circuit instances), {\tt Boeblingen} ($N = 20$, 10 circuit instances), and {\tt Cairo} ($N = 27$, 16 circuit instances). Their coupling graphs are shown in \Cref{fig:benchmark-devices}. Within this suite, the {\tt Melbourne} instances are generally easier than those of the other two devices, whereas {\tt Cairo} contains a larger share of long-running instances.

\begin{table}[htbp]
\centering
\footnotesize
\caption{Final configurations retained for the journal experiments.  \(P\),
\(S_0\), \(S_\ast\), and \(L\) are parts of the combinatorial bundle.  \(E\) denotes the engineering bundle.}
\label{tab:configurations}
\begin{tabularx}{\columnwidth}{@{}l l X@{}}
\toprule
Label & Components & Role \\
\midrule
\(H\) & \HHB{} & the bound used in~\citep{valois2026} \\
\(G\) & \GLB{} & the baseline bound used in this work \\
\(GP\) & \GLB{} + \(P\) & assigned-cost filtering \\
\(GPS_0\) & \(GP\) + \(S_0\) & root physical-orbit pruning \\
\(GPS_\ast\) & \(GP\) + \(S_\ast\) & root and prefix symmetry \\
\(GPS_\ast L\) & \(GPS_\ast\) + \(L\) & conditioned parent-\LAP{} screening \\
\(GPS_\ast LE\) & \(GPS_\ast L\) + \(E\) &
final implementation with reusable device-profiles \\
\bottomrule
\end{tabularx}
\end{table}
 
\paragraph{\underline{Baseline and configurations}}  \Cref{tab:configurations} presents the configurations used in the experiment. Configuration \(H\) is the \HHB{} kernel of~\citep{valois2026}. Configuration \(G\) uses the \GLB{} path in the same codebase and is the main baseline because it is faster than $H$ in our experiments. Each subsequent row adds one tool to the preceding one; the last row, \(GPS_\ast LE\), is the final configuration. 

\paragraph{\underline{Execution control and reproducibility}}
All experiments used one Chapel locale with the Qthreads runtime. Processes and
their runtime threads were pinned explicitly to physical cores using
\texttt{taskset}; SMT was not used. All the experiments used one solver process. We did not impose an explicit NUMA memory-binding or interleaving policy; consequently, memory
placement followed the OS and Chapel runtime's first-touch
behavior. Each exact-solver configuration was executed five times, and the mean execution time is given in the tables. For all exact-solver configuration/circuit runs used in this work, either all five runs completed, or all timed out. We observed less than $15\%$ coefficient of variation (i.e., $stdev$/$mean$) values when the mean is less than 20 seconds and less than $10\%$ when it is more. It also tends to decrease with increasing solver time.   

\subsection{Single-thread ablation study on the combinatorial bundle}

The first experiment is a single-thread ablation study on the 21 {\tt Melbourne} and 10 {\tt Boeblingen} instances.  For this study, following the experimental setting of Valois~\emph{et al.}~\citep{valois2026}, the pruning cutoff for each run is set to \(z^\star+1\), where \(z^\star\) is the optimal value. This is an oracle cutoff used to isolate bound and node-processing costs, not a feasible incumbent generated by the algorithm. Some jobs are timed out under the 21,600-second restriction. Timeouts of \(G\) are excluded from geometric means. Every completed run returned the reference optimum from~\citep{valois2026}. 

The cost of building the {\tt Melbourne} and {\tt Boeblingen} device profiles is not charged to the search time of a circuit; however, their loading times are included in the reported end-to-end times. The profile built for {\tt Melbourne} took only $2.6$ seconds, whereas the {\tt Boeblingen} profile took $46.0$ seconds. Note that these are single-thread runtimes, and profile generation is straightforward to parallelise.

The results of the experiment are presented in \Cref{tab:suite-a-melbourne,tab:suite-a-rd}. On {\tt Melbourne},
\(P\) changes the geometric mean by only \(1.07\times\); root-orbit pruning is the first decisive step, taking the cumulative speedup over \(G\) to \(2.17\times\).  Prefix symmetry is neutral on these small cases, after which \(L\) and \(E\) raise the final gain to \(2.98\times\).  On the 20 {\tt Melbourne} cases completed by both \(H\) and the final configuration, the corresponding gain over \(H\) is \(8.40\times\).  The excluded case is also the clearest individual win: \(H\) times out on \cfg{10\_qft}, while \(GPS_\ast LE\) finishes in 19.9 seconds.

The harder {\tt Boeblingen} cases expose the cumulative value of the full chain. Over RD11--RD16, root symmetry improves \(GP\) by \(3.04\times\), extending it to the prefix stabilizer adds \(1.27\times\), conditioned parent-\LAP{} screening adds \(1.80\times\) by avoiding many child-bound evaluations, and engineering reuse adds another \(1.88\times\).  The
result is a \(13.27\times\) geometric-mean speedup over \(G\) and a \(375\times\) speedup over \(H\). The final configuration is the fastest configuration on every {\tt Boeblingen} case that it completes.

\begin{table}[htbp]
\centering
\scriptsize
\setlength{\tabcolsep}{1.1pt}
\renewcommand{\arraystretch}{0.92}
\caption{Ablation study for the combinatorial bundle (\emph{End-to-end, sequential runtimes}):  \(\mathrm{TO}\) denotes the six-hour timeout in all the five runs. Each GM row reports the geometric mean of \(t_G/t_X\) for the configuration $X$; {\tt Melbourne} geometric means use all 21 cases; its \(H\) entry is an upper bound obtained by using the 21,600 seconds timeout. {\tt Boeblingen} geometric means use RD11--RD16 which \(G\) completes.}
\label{tab:suite-a}
\label{tab:suite-a-melbourne}
\label{tab:suite-a-rd}
\scalebox{1.15}{
\begin{tabular}{@{}cl|r|rrrrr|r@{}}
\toprule
D. & Instance & \(H\) & \(G\) & \(GP\) & \(GPS_0\) & \(GPS_\ast\) &
\(GPS_\ast L\) & \(GPS_\ast LE\)\\
\midrule
\multirow{22}{*}{\rotatebox[origin=c]{90}{{\tt Melbourne}}}
 & 10\_qft     & TO    & 176.25 & 140.32 & 35.87 & 35.62 & 30.08 & 19.87\\
 & 10\_sqn     & 2.41  & 0.33   & 0.27   & 0.25  & 0.26  & 0.26  & 0.26\\
 & 10\_sym9    & 1.18  & 0.29   & 0.26   & 0.26  & 0.26  & 0.26  & 0.26\\
 & 11\_sym9    & 2.42  & 0.55   & 0.54   & 0.29  & 0.28  & 0.26  & 0.26\\
 & 11\_wim     & 2.72  & 0.26   & 0.26   & 0.26  & 0.26  & 0.26  & 0.26\\
 & 11\_z4      & 8.09  & 0.73   & 0.62   & 0.28  & 0.29  & 0.27  & 0.26\\
 & 12\_cycle10 & 3.78  & 1.11   & 1.11   & 0.46  & 0.42  & 0.34  & 0.30\\
 & 12\_rd84    & 1.73  & 0.77   & 0.71   & 0.31  & 0.38  & 0.29  & 0.26\\
 & 12\_sym9    & 5.65  & 0.26   & 0.26   & 0.26  & 0.26  & 0.26  & 0.26\\
 & 13\_dist    & 5.01  & 1.93   & 1.86   & 0.61  & 0.60  & 0.41  & 0.32\\
 & 13\_radd    & 3.56  & 0.78   & 0.75   & 0.32  & 0.38  & 0.28  & 0.27\\
 & 13\_root    & 5.07  & 1.86   & 1.81   & 0.60  & 0.57  & 0.41  & 0.30\\
 & 14\_clip    & 14.47 & 11.95  & 11.63  & 3.07  & 3.09  & 1.97  & 1.23\\
 & 14\_cm42a   & 1.65  & 0.33   & 0.27   & 0.28  & 0.27  & 0.27  & 0.28\\
 & 14\_cm85a   & 7.40  & 7.67   & 7.58   & 2.03  & 2.10  & 1.20  & 0.82\\
 & 15\_co14    & 1.04  & 0.64   & 0.64   & 0.32  & 0.35  & 0.37  & 0.31\\
 & 15\_misex1  & 2.10  & 0.87   & 0.91   & 0.41  & 0.34  & 0.29  & 0.29\\
 & 15\_sqrt7   & 4.29  & 1.57   & 1.30   & 0.47  & 0.52  & 0.38  & 0.31\\
 & 16\_inc     & 1.67  & 0.90   & 0.89   & 0.41  & 0.36  & 0.28  & 0.33\\
 & 16\_ising   & 0.27  & 0.30   & 0.28   & 0.27  & 0.30  & 0.32  & 0.27\\
 & 16\_mlp4    & 24.52 & 72.05  & 71.41  & 17.90 & 17.88 & 12.48 & 7.19\\
\cmidrule(l){2-9}
 & \textbf{GM \(t_G/t_X\)} & \(\le 0.282\) & 1.00 & 1.07 & 2.17 &
   2.15 & 2.59 & \textbf{2.98}\\
\midrule
\multirow{11}{*}{\rotatebox[origin=c]{90}{{\tt Boeblingen}}}
 & RD11 & 1,434.7  & 7.7    & 7.6    & 2.5    & 2.1     & 1.1    & 0.7\\
 & RD12 & 2,220.9  & 30.5   & 30.5   & 9.7    & 7.9     & 3.7    & 1.9\\
 & RD13 & 1,619.4  & 30.4   & 30.0   & 11.1   & 7.9     & 3.8    & 2.0\\
 & RD14 & 5,959.6  & 326.0  & 317.0  & 105.8  & 79.8    & 45.1   & 22.4\\
 & RD15 & 13,840.9 & 1,657.8 & 1,653.4 & 491.9  & 435.8   & 267.3  & 133.6\\
 & RD16 & 8,253.2  & 1,790.5 & 1,738.8 & 558.5  & 421.2   & 309.6  & 168.7\\
 & RD17 & TO      & TO     & TO     & 9,111.7 & 7,695.9  & 7,051.2 & 3,908.0\\
 & RD18 & TO      & TO     & TO     &  8,468.1 & 5,499.4 & 4,517.2 & 2,653.5\\
 & RD19 & TO       & TO      & TO      & TO      & TO       & TO      & TO\\
 & RD20 & TO      & TO     & TO     & TO     & TO      & TO     & 14,320.6\\
\cmidrule(l){2-9}
 & \textbf{GM \(t_G/t_X\)} & 0.04 & 1.00 & 1.02 & 3.09 & 3.92 &
   7.07 & \textbf{13.27}\\
\bottomrule
\end{tabular}
}
\end{table}
In terms of coverage, both \(G\) and \(H\)
complete only RD11--RD16, whereas symmetry brings RD17 and RD18 below the
six-hour limit and \(GPS_\ast LE\) additionally completes RD20.  Thus the
final configuration solves nine of ten {\tt Boeblingen} cases;
RD19 runs were unsuccessful for all configurations. Note that this experiment uses a single core of a single machine.

\Cref{tab:suite-a-trees} presents the impact of the configurations on pruning the tree. Relative to \(G\), \(S_0\) reduces the geometric-mean tree size by
\(3.78\times\) on {\tt Melbourne} and \(3.04\times\) on RD11--RD16. The full prefix reduction \(S_\ast\) yields corresponding reductions of
\(3.78\times\) and \(3.96\times\), respectively.  Thus \(S_\ast\) matches
\(S_0\) on {\tt Melbourne}, but provides a further \(1.30\times\) reduction
over \(S_0\) on RD11--RD16. Both devices have \(|\Gamma|=4\), but their point stabilizers differ.  The
{\tt Melbourne} root orbits are
\[
 \{0,7,8,15\},\quad \{1,6,9,14\},\quad
 \{2,5,10,13\},\quad \{3,4,11,12\}.
\]
Every orbit has size four. 
Once the first vertex is occupied, no  symmetry remains below the
root; consequently, \(GPS_0\) and \(GPS_\ast\) necessarily retain the same
{\tt Melbourne} tree.

\begin{table}[htbp]
\centering
\scriptsize
\setlength{\tabcolsep}{2pt}
\renewcommand{\arraystretch}{0.94}
\caption{Ablation study for the combinatorial bundle (\emph{Size of the explored tree}): \(\mathrm{TO}\)
denotes the six-hour timeout. The tree size describes the retained search tree; each GM row reports
the geometric mean of \(T_G/T_X\), where \(T_X\) is the tree size of
configuration \(X\). {\tt Melbourne} geometric means use all 21 cases and {\tt Boeblingen} geometric means use
RD11--RD16. The final column combines \(GPS_\ast\), \(GPS_\ast L\), and
\(GPS_\ast LE\), whose counts agree wherever all three are complete;
\(\dagger\) implies that the result was available only for \(GPS_\ast LE\).}
\label{tab:suite-a-trees}
\scalebox{1.12} {
\begin{tabular}{@{}cl|rrrr|r@{}}
\toprule
D. & Instance & \(H\) & \(G\) & \(GP\) & \(GPS_0\) &
\shortstack{\(GPS_\ast\)\\\(GPS_\ast L\)\\\(GPS_\ast LE\)}\\
\midrule
\multirow{22}{*}{\rotatebox[origin=c]{90}{{\tt Melbourne}}}
 & 10\_qft & TO & 6,446,808 & 6,446,807 & 1,611,707 & 1,611,707\\
 & 10\_sqn & 206 & 1,605 & 1,605 & 405 & 405\\
 & 10\_sym9 & 208 & 1,525 & 1,524 & 384 & 384\\
 & 11\_sym9 & 208 & 5,810 & 5,810 & 1,454 & 1,454\\
 & 11\_wim & 538 & 512 & 512 & 134 & 134\\
 & 11\_z4 & 617 & 8,105 & 8,105 & 2,030 & 2,030\\
 & 12\_cycle10 & 332 & 14,953 & 14,953 & 3,751 & 3,751\\
 & 12\_rd84 & 170 & 7,830 & 7,830 & 1,962 & 1,962\\
 & 12\_sym9 & 716 & 401 & 401 & 110 & 110\\
 & 13\_dist & 420 & 23,155 & 23,155 & 5,791 & 5,791\\
 & 13\_radd & 348 & 6,522 & 6,522 & 1,641 & 1,641\\
 & 13\_root & 389 & 23,073 & 23,073 & 5,775 & 5,775\\
 & 14\_clip & 1,168 & 180,928 & 180,928 & 45,265 & 45,265\\
 & 14\_cm42a & 693 & 1,079 & 1,078 & 280 & 280\\
 & 14\_cm85a & 566 & 97,496 & 97,496 & 24,383 & 24,383\\
 & 15\_co14 & 45 & 3,122 & 3,122 & 782 & 782\\
 & 15\_misex1 & 148 & 4,794 & 4,794 & 1,206 & 1,206\\
 & 15\_sqrt7 & 234 & 17,671 & 17,671 & 4,432 & 4,432\\
 & 16\_inc & 114 & 4,415 & 4,415 & 1,109 & 1,109\\
 & 16\_ising & 35 & 36 & 36 & 24 & 24\\
 & 16\_mlp4 & 1,935 & 1,003,775 & 1,003,775 & 251,183 & 251,183\\
\cmidrule(l){2-7}
 & \textbf{GM \(T_G/T_X\)} & --- & 1.00 & 1.00 & \textbf{3.78} & 3.78\\
\midrule
\multirow{11}{*}{\rotatebox[origin=c]{90}{{\tt Boeblingen}}}
 & RD11 & 93,628 & 71,366 & 71,366 & 22,497 & 18,071\\
 & RD12 & 114,348 & 296,358 & 296,358 & 92,587 & 74,811\\
 & RD13 & 66,388 & 237,950 & 237,950 & 88,219 & 60,798\\
 & RD14 & 324,179 & 3,178,349 & 3,178,349 & 1,075,823 & 799,606\\
 & RD15 & 685,293 & 14,719,034 & 14,719,034 & 4,459,393 & 3,699,591\\
 & RD16 & 411,128 & 13,254,737 & 13,254,737 & 4,455,115 & 3,343,208\\
 & RD17 & TO & TO & TO & 75,310,690 & 63,313,572\\
 & RD18 & TO & TO & TO & 43,257,340 & 34,189,708\\
 & RD19 & TO & TO & TO & TO & TO\\
 & RD20 & TO & TO & TO & TO & 147,633,049\(^{\dagger}\)\\
\cmidrule(l){2-7}
 & \textbf{GM \(T_G/T_X\)} & --- & 1.00 & 1.00 & \textbf{3.04} & \textbf{3.96}\\
\bottomrule
\end{tabular}
}
\end{table}

The {\tt Boeblingen} root orbits are
\[
\begin{gathered}
 \{0,4,15,19\},\ \{1,3,16,18\},\ \{2,17\},\\
 \{5,9,10,14\},\ \{6,8,11,13\},\ \{7,12\}.
\end{gathered}
\]
For \(p\in\{2,7,12,17\}\), the orbit has size two and therefore
\(|\Gamma_p|=4/2=2\).  A nonidentity left--right reflection fixes these
middle-column vertices while exchanging, for example,
\(0\leftrightarrow4\), \(1\leftrightarrow3\),
\(5\leftrightarrow9\), and \(6\leftrightarrow8\).  If the first assignment
uses a fixed vertex, this reflection survives and \(S_\ast\) keeps one next
candidate from each exchanged pair. This raises the cumulative {\tt Boeblingen} tree reduction
from \(3.04\times\) to \(3.96\times\), a further \(1.30\times\).

The almost identical \(G\) and \(GP\) tree sizes are expected.  The
filter \(P\) tests \(C_{A,A}\), whereas the subsequent \GLB{} is
\(C_{A,A}+\lap(\mathbf G)\geq C_{A,A}\); hence every candidate rejected by
\(P\) would also be rejected by \GLB{} before entering the DFS pool.  Thus
\(P\) does not reduce the retained tree, but can avoid constructing and
solving the corresponding assignment problem. The isolated one-node
differences on {\tt Melbourne} instances result from complete equal-cost allocations
being intercepted before insertion.

The configurations \(GPS_\ast\), \(GPS_\ast L\), and
\(GPS_\ast LE\) have the same tree sizes. Without \(L\), the rejection process is
\[
\begin{aligned}
 \text{candidate}&\longrightarrow\text{construct child}
   \longrightarrow\text{build child GLB}\\
 &\longrightarrow\text{Hungarian solve}
   \longrightarrow\text{reject child}.
\end{aligned}
\]
The rejected child is never inserted into the DFS pool and is therefore not
included in the solver's reported tree size.  With \(L\),
\Cref{fig:forced-lap}(a)--(c) depicts the corresponding path:
\[
\begin{aligned}
 \text{parent LAP cert.}&\longrightarrow\text{force candidate}
 &\longrightarrow\text{reject early}.
\end{aligned}
\]
By \Cref{prop:l-dominated}, every candidate rejected by \(L\) would also be rejected by the rebuilt child \GLB{} at a fixed cutoff. The two versions, therefore, insert the same nodes into the DFS pool, while \(L\) avoids constructing and bounding many children that will be rejected. For instance, on RD15, we have the following table:
\begin{center}
\footnotesize
\begin{tabular}{@{}lrr@{}}
\toprule
Configuration & Explored tree & GLB calls\\
\midrule
Without \(L\) & 3,699,591 & 47,499,942\\
With \(L\)    & 3,699,591 & 27,609,039\\
\bottomrule
\end{tabular}
\end{center}
Thus \(L\) screened approximately 19.9 million candidates before their child \GLB{} computations.

Finally, the \(H\) column
shows why tree size alone does not predict time: its stronger but more
expensive bound often retains far fewer nodes than \(G\), yet the wall-time
comparison in \Cref{tab:suite-a} favours the cheaper bound.

\subsection{Engineering bundle ablation study}

\Cref{tab:engineering-results} dissects the performance gain without changing the retained search tree or the decisions based on combinatorial evaluations. Every configuration expands the same tree. The \(1.83\times\) geometric-mean gain of the full bundle on RD15--RD18, therefore, comes from cheaper state maintenance, data access, and certificate reuse rather than additional pruning. 

\begin{table}[htbp]
\centering
\scriptsize
\setlength{\tabcolsep}{2.2pt}
\caption{Engineering bundle dissection on RD15--RD20 (\emph{End-to-end, sequential runtimes}): \(B=GPS_\ast L\), \(F=GPS_\ast LE\),
\(I=\mathrm{INC}\), \(D=\mathrm{DDOT}\), \(A=\mathrm{PAU}\),
\(C=\mathrm{CERT}\), and \(T=\mathrm{TINY}\).  The \(B+IDA\),
\(B+IDC\), and \(B+IDT\) are separate extensions of \(B+ID\).
The last row is the geometric-mean speedup relative to \(B\).}
\label{tab:engineering-results}
\scalebox{1.15}{
\begin{tabular}{@{}lrrrrrrr@{}}
\toprule
Instance & \(B\) & \(B+I\) & \(B+ID\) & \(B+IDA\) &
\(B+IDC\) & \(B+IDT\) & \(F\)\\
\midrule
RD15 & 267.3  & 264.3  & 219.6  & 156.5  & 197.6  & 220.2  & 133.6\\
RD16 & 309.6  & 296.1  & 252.3  & 180.1  & 236.3  & 246.1  & 168.7\\
RD17 & 7051.2 & 6138.4 & 5256.9 & 3496.6 & 5053.9 & 5409.0 & 3908.0\\
RD18 & 4517.2 & 4172.1 & 3565.9 & 2548.8 & 3367.9 & 3584.6 & 2653.5\\
RD19 & TO      & TO      & TO      & TO      & TO      & TO      & TO\\
RD20 & TO      & TO      & TO      & TO      & TO      & TO      & 14320.6\\
\midrule
 \textbf{GM \(t_B/t_X\)}  & 1.00 & 1.07 & 1.26 & 1.80 & 1.35 & 1.26 & \textbf{1.83}\\
\bottomrule
\end{tabular}
}
\end{table}

Incremental state and device dots together give a \(1.26\times\) gain over
the unengineered base.  Parent assigned--unassigned reuse ($IDA$) is the strongest isolated extension, raising this to \(1.80\times\); certificate reuse ($IDC$) is smaller, and the tiny kernels ($IDT$) are neutral relative to \(ID\) in geometric mean.  Overall, the full bundle is only \(1.02\times\) faster than the $IDA$ branch on average and is slower on RD17 and RD18.  Nevertheless, it wins on RD15 and RD16 and is the only engineering configuration to complete RD20. We therefore use the whole \(E\) bundle in the final configuration for its overall robustness and coverage, not because every component is individually or monotonically beneficial.

\subsection{Practical incumbent heuristics}

For completeness, before continuing the multi-threaded experiments, we present a short evaluation of the three increasingly expensive incumbent heuristics in \Cref{tab:heuristic-results} on all 47 benchmark
instances. One-/two-exchange descent more than halves the median gap while retaining millisecond cost. The final heuristic with a 30-second budget finds an optimum in 35 cases (all 21 {\tt{Melbourne}} circuits, nine of ten {\tt{Boeblingen}} circuits, and five of 16 {\tt{Cairo}} circuits) and reduces the mean gap from 11.27\% to 0.45\%. In the rest of this section, we use it to set the initial incumbent, thereby avoiding the oracle value \(z^\star\) in the practical solver configuration.

\begin{table}[htbp]
\centering
\scriptsize
\setlength{\tabcolsep}{4.2pt}
\caption{Feasible-incumbent quality on 47 instances.  Gaps are relative to the known optimum.}
\label{tab:heuristic-results}
\scalebox{1.2}{
\begin{tabular}{@{}lrrrr@{}}
\toprule
Heuristic & Opt. & \shortstack{Median\\gap} & \shortstack{Mean\\gap} &
  \shortstack{Mean\\time}\\
\midrule
{\sc GreedyAllocation}  & 4/47  & 8.60\% & 11.27\% & 7.4 ms \\
 \(+\) {\sc OneTwoDescent} & 9/47  & 3.95\% & 4.88\% & 14.3 ms \\
 \(+\) {\sc BudgetedSearch} (30 s) & 35/47 & 0.00\% & 0.45\% & 30.0 s\\
\bottomrule
\end{tabular}
}
\end{table}

\subsection{Shared-memory scaling}

The scaling experiment runs one \(GPS_\ast LE\) process at a time on RD14--RD18 with 1, 2, 4, 8, 16, and 32 Chapel threads. All 150 jobs were completed and returned the optimum value.  The geometric-mean search speedup rises monotonically from \(1.68\times\) at two threads to \(27.57\times\) at 32 threads, corresponding to 86\% parallel efficiency.  

\begin{table}[htbp]
\centering
\scriptsize
\setlength{\tabcolsep}{3.2pt}
\caption{{Strong scaling.  \(T_1\) is one-thread search time in
seconds; \(S_p=T_1/T_p\).  \(K_0\) is the exact upper bound returned by the
30-second heuristic; the last column includes that same heuristic cost in
both numerator and denominator.}}
\label{tab:scaling-results}
\scalebox{1.15}{
\begin{tabular}{@{}lrrrrrrrr@{}}
\toprule
Instance & {\(K_0\)} & \(T_1\) & \(S_2\) & \(S_4\) & \(S_8\) & \(S_{16}\) &
\(S_{32}\) & \shortstack{End-to-end\\\(S_{32}\)}\\
\midrule
RD14 & {428}  & 25.0   & 1.66 & 3.86 & 6.31 & 12.68 & 25.98 & 1.77\\
RD15 & {682}  & 146.8  & 1.53 & 3.69 & 6.55 & 13.23 & 30.41 & 5.03\\
RD16 & {822}  & 164.8  & 1.57 & 2.83 & 6.27 & 13.84 & 29.67 & 5.43\\
RD17 & {840}  & 5553.2 & 2.17 & 4.11 & 6.90 & 12.30 & 20.75 & 18.73\\
RD18 & {1140} & 2822.8 & 1.57 & 4.73 & 8.99 & 17.41 & 32.74 & 24.47\\
\midrule
GM & {---} & --- & 1.68 & 3.79 & 6.94 & 13.78 & \textbf{27.57} & 7.40\\
\bottomrule
\end{tabular}
}
\end{table}

Four instances retain the same expanded-tree count at every thread
count.  RD17 is the exception: its tree grows from 88.5 to 121.8 million
nodes, explaining its lower parallel efficiency.  The fixed heuristic also
dominates the shortest cases: the search-only gain on RD14 is
\(25.98\times\), but the practical end-to-end gain is \(1.77\times\).
Across all five cases, the 32-thread end-to-end geometric-mean speedup remains \(7.40\times\).  

\subsection{Final 60-core evaluation}

The final experiment uses one 60-thread \(GPS_\ast LE\) process at a time on RD15--RD20 and all 16 {\tt Cairo} circuits. The full device profiles have been used for thıs experiment. \Cref{tab:final60-results} presents the results of this experiment. With 60 threads, all runs are finished, each in less than half an hour, and all returned the reference optimum.

\begin{table}[htbp]
\centering
\scriptsize
\setlength{\tabcolsep}{3.5pt}
\caption{60-thread results.  \(K_0\) is the exact feasible upper bound returned
by the 30-second heuristic and supplied to the solver; end-to-end time includes
the heuristic, device-profile loading, and exact solver wall time. The reusable
device-artifact construction is excluded.}
\label{tab:final60-results}
\scalebox{1.15}{
\begin{tabular}{@{}clrrrr@{}}
\toprule
D. & Instance & \(n\) & \(z^\star\) & {\(K_0\)} &
  \shortstack{Search /\\end-to-end (s)}\\
\midrule
\multirow{6}{*}{\rotatebox[origin=c]{90}{{\tt Boeblingen}}}
 & RD15 & 15 & 682  & {682}  & 2.5 / 32.9\\
 & RD16 & 16 & 822  & {822}  & 2.7 / 33.1\\
 & RD17 & 17 & 826  & {840}  & 113.5 / 143.9\\
 & RD18 & 18 & 1140 & {1140} & 46.6 / 77.0\\
 & RD19 & 19 & 1400 & {1400} & 1150.5 / 1180.9\\
 & RD20 & 20 & 1582 & {1582} & 301.6 / 332.1\\
\midrule
\multirow{16}{*}{\rotatebox[origin=c]{90}{{\tt Cairo}}}
 & 5xp1       & 17 & 3034   & {3036}   & 0.4 / 36.2\\
 & ryy6       & 17 & 14942  & {14942}  & 53.0 / 88.9\\
 & alu3       & 18 & 8306   & {8306}   & 6.3 / 42.2\\
 & sqrt6      & 18 & 2222   & {2232}   & 0.3 / 36.2\\
 & add6       & 19 & 15408  & {15408}  & 0.2 / 36.8\\
 & cmb        & 20 & 4428   & {4448}   & 1.1 / 37.4\\
 & ex1010     & 20 & 224476 & {224732} & 0.1 / 36.6\\
 & decod      & 21 & 3644   & {3668}   & 0.2 / 36.7\\
 & dk17       & 21 & 3600   & {3632}   & 0.4 / 37.0\\
 & apla       & 22 & 8004   & {8004}   & 0.6 / 37.2\\
 & cm105a     & 22 & 1780   & {1908}   & 1.4 / 37.9\\
 & 63mod4096  & 23 & 52     & {52}     & 0.2 / 36.5\\
 & arb8       & 24 & 1880   & {1908}   & 1737.1 / 1774.3\\
 & cu         & 25 & 2652   & {2688}   & 30.2 / 69.4\\
 & rd73       & 25 & 654    & {672}    & 522.5 / 561.1\\
 & in0        & 26 & 44912  & {46740}  & 57.4 / 96.1\\
\bottomrule
\end{tabular}
}
\end{table}

The slowest end-to-end run is {\tt Cairo} \cfg{arb8} at 29.6 minutes on average (all runs are completed). Twenty of 22 were completed within ten minutes end-to-end, on average, and 18 were completed within five minutes. On RD15--RD18, which also appear in the single-thread experiment, 60 cores give a \(56.8\times\) geometric-mean search speedup. Even RD19, which none of the single-thread configurations in \Cref{tab:suite-a} completed within six hours, is certified in 19.7 minutes end to end on average.

Valois \emph{et al.} evaluated the same {\tt Cairo} cases on 64 nodes with 128 cores per node, 8192 cores in total~\citep{valois2026}.  The present solver certifies the same 16 optima with 60 cores in one shared-memory server. Since processors differ, this is a {\em resource scale} comparison rather than a cross-machine speedup. It nevertheless establishes that at this benchmark scale, graph-aware reductions help to move exact allocation into the reach of one conventional shared-memory machine. 

\subsection{The impact of device cache on {\tt Cairo}}
\label{sec:cairo-cache-comparison}

The {\tt Cairo} circuits in \Cref{tab:final60-results} use a full-profile cache. As mentioned before, profile construction is timed once and excluded from per-circuit mapping time because the device is fixed and the file is reusable. Here, we focus on two distinct questions: whether indexed profiles accelerate the search, and whether that saving is large enough to repay the load and the per-circuit setup.  
The experiment provides a clear separation between kernel and
end-to-end effects. Every cached/on-demand run returns the expected optimum. Search time improves for 13 instances, by \(1.54\times\) on geometric average. Parallel scheduling changes the explored count slightly: six pairs are identical, and the geometric mean of cached/on-demand tree size is 0.994. Normalising search time by explored nodes still gives a \(1.54\times\) geometric-mean kernel gain.
More importantly, for the long cases, the complete end-to-end run improves by \(1.92\times\) on \cfg{arb8}, \(1.55\times\) on \cfg{rd73}, and \(1.75\times\) on \cfg{ryy6}. \Cref{tab:cairo-cache-timings} retains
all per-instance timings needed to expose where the cache wins and loses.

\begin{table}[htbp]
\centering
\scriptsize
\setlength{\tabcolsep}{1.5pt}
\renewcommand{\arraystretch}{0.92}
\caption{Per-instance {\tt Cairo} device-profile timings for the 60-thread
\(GPS_\ast LE\) configuration.  \(t_s\) is exact-search time and \(t_w\)
is solver wall time, including the device artifact load and all
circuit-dependent setup, but excluding the common 30-second incumbent
heuristic. \(S=t^{\mathrm{on}}/t^{\mathrm{cache}}\); all times are seconds, and all cached runs return the same optimum.}
\label{tab:cairo-cache-timings}
\scalebox{1.2} {
\begin{tabular}{@{}lrrrrrr@{}}
\toprule
& \multicolumn{3}{c}{Search} & \multicolumn{3}{c}{Wall}\\
\cmidrule(lr){2-4}\cmidrule(l){5-7}
Instance & On-demand & Cache & \(S_s\) & On-demand & Cache & \(S_w\)\\
\midrule
5xp1       & 0.8    & 0.4    & 2.00 & 1.1    & 6.2    & 0.18 \\
ryy6       & 102.8  & 53.0   & 1.94 & 103.1  & 58.8   & 1.75 \\
alu3       & 12.2   & 6.3    & 1.94 & 12.5   & 12.2   & 1.02 \\
sqrt6      & 0.2    & 0.3    & 0.67 & 0.6    & 6.2    & 0.10 \\
add6       & 0.4    & 0.2    & 2.00 & 0.8    & 6.8    & 0.12 \\
cmb        & 2.2    & 1.1    & 2.00 & 2.6    & 7.4    & 0.35 \\
ex1010     & 0.1    & 0.1    & 1.00 & 0.5    & 6.5    & 0.08 \\
decod      & 0.2    & 0.1    & 2.00 & 0.6    & 6.7    & 0.09 \\
dk17       & 0.8    & 0.4    & 2.00 & 1.2    & 6.9    & 0.17 \\
apla       & 1.1    & 0.6    & 1.83 & 1.5    & 7.1    & 0.21 \\
cm105a     & 2.7    & 1.4    & 1.93 & 3.1    & 7.8    & 0.40 \\
63mod4096  & 0.1    & 0.2    & 0.50 & 0.5    & 6.5    & 0.08 \\
arb8       & 3343.0 & 1737.1 & 1.92 & 3343.4 & 1744.2 & 1.92 \\
cu         & 46.2   & 30.2   & 1.53 & 46.5   & 39.3   & 1.18 \\
rd73       & 822.8  & 522.5  & 1.57 & 823.0  & 531.1  & 1.55 \\
in0        & 91.9   & 57.4   & 1.60 & 92.2   & 65.8   & 1.40 \\
\midrule
Geometric mean
& -- & -- & \textbf{1.54}
& -- & -- & \textbf{0.35} \\
\bottomrule
\end{tabular}
}
\end{table}

The negative~(less than 1$\times$ speedup) wall-time result on smaller (and hence easier) circuits is informative. The on-demand backend has only \(0.28\)--\(0.43\) seconds of non-search setup, whereas full-cache non-search setup---including artifact loading, validation, dot computations, and startup---takes \(5.79\)--\(9.10\) seconds per solver process. Consequently, the cache is faster in solver wall time on six of 16 cases but \(2.86\times\) slower geometrically over the complete experiment, despite its faster search kernel. Building the reusable 7.81-GB device artifact took 113.9 seconds (with 60 threads) and is intentionally excluded from both arms. 

A case-by-case, equal-weight comparison hides the batch
benefit. The sum of the 16 solver-wall times in~\Cref{tab:cairo-cache-timings} falls from 4433.2 to 2519.5 seconds. Adding the common 30-second heuristic gives end-to-end totals of 4913.2 versus 2999.5 seconds; 
charging the 113.9-second cache construction once raises the
cached total to 3113.4 seconds, for a $1.58\times$ batch speedup.
Overall, profiling the device pays for sufficiently long searches, while exact on-demand processing remains preferable for short per-circuit runs. A persistent service that retained the cache across circuits could remove the per-process-load penalty, but was left outside in our experiment.

\FloatBarrier
\section{Related work}
\label{sec:related}

Foundational qubit-allocation work formalises initial placement and relates it to graph embedding, subgraph isomorphism, and token swapping~\citep{siraichi2018,siraichi2019}. The closest prior work is the exact allocation framework of Valois \emph{et al.} \citep{valois2026}. It establishes the qubit-allocation
formulation, a strong \HHB{}-based branch-and-bound procedure, and scalable Chapel search. Our work is complementary to~\citep{valois2026}; we have built all the proposed techniques on top of a reduced version of the codebase of Valois~et~al.~\citep{valois2026}, and we kept \HHB{} as a control and \GLB{} as a cheaper alternative, while studying graph-aware reductions. Hence, although we do not present a new parallel search framework, we have shown that our stronger graph-aware node processing can postpone the point at which distributed execution becomes necessary. For instance, their {\tt Cairo} experiment used 64 nodes of 128 cores; ours uses one conventional 60-core shared-memory server and obtains the same certification faster. Furthermore, since we have started with an existing codebase, which is highly scalable on a distributed compute cluster, our solver can be adapted to the distributed setting and solve much harder instances. 

Exact assignment-and-routing formulations optimise richer dynamic objectives than the static one studied here. Integer-programming and exact routing approaches, for example, model routing choices and circuit evolution explicitly~\citep{nannicini2022,wagner2023}. They define the scope boundary of the present work rather than direct runtime baselines.

The two baseline engines also come directly from the \QAP{} literature.
Gilmore and Lawler reduce row-wise quadratic costs to a linear assignment
\citep{gilmore1962,lawler1963}; Anstreicher reviews this classical construction and its opposite-order scalar-product interpretation~\citep{anstreicher2003}. The Hahn--Grant dual procedure retains a much larger four-index complementary-cost representation~\citep{hahn1998,hahn1998bnb}. General \QAP{} surveys and the fewer-objects-than-locations variant provide broader context~\citep{loiola2007,kaibel1998}. Section~\ref{sec:background} makes
their details and complexity explicit, since a comparison is essential to understanding the benefits of the proposed techniques. Note that neither \GLB{} nor \HHB{} is claimed as a contribution of this paper. 

Zhu \emph{et al.} are the direct prior work for the symmetry
component. Their exact \QAP{} branch-and-bound algorithm tests
structural symmetry among the remaining physical qubits and eliminates
redundant branches, following the partial-node symmetry test of Mautor
and Roucairol~\citep{zhu2020,mautor1994}. In the notation of this paper,
the underlying exact equivalence condition is the orbit relation of the
subgroup of \(\Aut(D)\) that fixes the occupied physical vertices
pointwise. In the broader \QAP{} literature, symmetry, equivalent facilities, and
representative branching are well-established reductions
\citep{fischetti2012,ostrowski2011,mckay1998}. Group-theoretic
restrictions and subarchitecture selection have also been used to reduce
optimal quantum-mapping search spaces
\citep{burgholzer2022,peham2023}. Our contribution with \(S_\ast\) is its incremental group-based realisation, its separate ablation against \(S_0\), and its integration with conditioned parent-\LAP{} screening and reusable device profiles.

The Hungarian algorithm is a standard assignment primitive
\citep{kuhn1955}, and alternating-path sensitivity is classical in ranked and constrained matching~\citep{chegireddy1987}. The proposed \(L\) uses a parent residual certificate to price every next allocation exactly within the parent relaxation, then exposes that computation as a separately measurable branch-and-bound operator. As \Cref{prop:l-dominated} shows, its main benefit is early avoidance of child-\GLB{} work, not a stronger retained-tree bound. 

Heuristic qubit mappers, including SABRE \citep{li2019sabre}, solve a
different but practically essential problem: they interleave mapping and
routing for circuits beyond the reach of exhaustive search. Exact static
allocation is not proposed as a substitute. Instead, it supplies verified
reference costs, controlled search trees, and graph-level evidence that can eventually inform heuristic policies. The benchmark is also aligned with reusable quantum-compilation evaluation practice
\citep{quetschlich2023}.

\section{Conclusion and future work}
\label{sec:conclusion}

This paper develops a focused graph-aware approach for exact
static qubit allocation.  Starting from a simple injection model, it
separates three exact mathematical ideas---unavoidable assigned cost,
device-graph symmetry that survives a partial assignment, and conditioned parent-\LAP{} values---from an engineering bundle that makes those certificates cheap to maintain. 

The sequential experiments show where the gain comes from.  Assigned-cost
filtering alone is inexpensive but has little runtime effect; symmetry
produces the first large retained-tree reduction, conditioned parent-\LAP{} values screen candidates before their child \GLB{} computations, and implementation reuse accelerates the same certified tree. Together, the operators give a \(13.27\times\) geometric-mean improvement over the existing \GLB{} path on RD11--RD16 and extend six-hour RD coverage from six to nine of ten cases.

Overall, shared-memory execution is shown to be sufficient for the benchmark
scale considered here. A heuristic-seeded 60-core process certifies all 22 final {\tt Boeblingen} and {\tt Cairo} cases within 29.6 minutes per instance on one machine. This complements the literature~\cite{valois2026}: their work shows that branch-and-bound can exploit a cluster, whereas ours shows that graph-aware reductions can make this process much faster and feasible on a shared-memory, multicore server for some instances. Broader circuit
families, topology-diversity studies, adaptive selective caching, transfer of exact certificates into dynamic-routing heuristics, and distributed execution beyond the single-machine frontier remain important future directions.

\section*{Data and code availability}

The benchmark matrices and source distribution are
publicly available at \url{https://github.com/kamerkaya/
StaticQubitAllocation}. The distribution includes the exact
solver, device-profile builders, and the one-instance incumbent
generator.

\section*{Funding}

This research did not receive any specific grant from funding agencies in the public, commercial, or not-for-profit sectors.

\section*{Declaration of competing interest}

The author declares that they have no known competing
financial interests or personal relationships that could have appeared to
influence the work reported in this paper.

\section*{CRediT authorship contribution statement}

Kamer Kaya: Conceptualization, Methodology, Software, Writing – original draft.

\section*{Declaration of generative AI and AI-assisted technologies in the
manuscript preparation process}

During the preparation of this work, the author used OpenAI ChatGPT and
Codex to support code development, experiment orchestration, and manuscript
editing.  The author reviewed and verified the resulting code, analyses,
references, and text, revised the material as needed, and takes full
responsibility for the content of the article.

\bibliographystyle{elsarticle-num}
\bibliography{main}
\balance
\end{document}